\documentclass[a4paper,UKenglish,cleveref,thm-restate]{lipics-v2021}

\hideLIPIcs

\title{Solvability of Approximate Agreement on Graphs and Simplicial Complexes}

\titlerunning{Solvability of Approximate Agreement on Graphs and Simplicial Complexes}

\author{Joel Rybicki}{Humboldt University of Berlin}{joel.rybicki@hu-berlin.de}{}{}
\author{Yaroslav Verbitsky}{Humboldt University of Berlin}{yaroslav.verbitsky@student.hu-berlin.de}{}{}

\authorrunning{J. Rybicki and Y. Verbitsky}

\Copyright{Joel Rybicki and Yaroslav Verbitsky}



\category{}

\relatedversion{}

\acknowledgements{Parts of this work are based on the Master's thesis of the second author \cite{verbitsky2026cliqueagreement}.}

\nolinenumbers

\newcommand{\K}{\mathcal{K}}
\newcommand{\A}{\mathcal{A}}

\DeclareMathOperator{\Bary}{Bary}

\DeclareMathOperator{\conv}{conv}

\DeclareMathOperator{\id}{id}

\DeclareMathOperator{\skel}{skel}

\newcommand{\R}{\mathds{R}}

\newcommand{\I}{\mathcal{I}}
\newcommand{\Ocal}{\mathcal{O}}

\newcommand{\B}{\mathcal{B}}
\newcommand{\C}{\mathcal{C}}

\newcommand{\he}{\simeq}

\newcommand{\Set}[2]{\{ #1 : #2 \}}
\newcommand{\set}[1]{\{ #1 \}}

\newcommand{\Deltava}{\Delta^{V(\A)}}

\newcommand{\skelalphA}{\skel^t\A}
\newcommand{\skelalphDeltava}{\skel^t\Deltava}
\newcommand{\DeltaPhisig}{\Delta^{V(\PhiC(\sigma))}}
\newcommand{\skelalphPhisig}{\skel^t\PhiC(\sigma)}
\newcommand{\skelalphDeltaPhisig}{\skel^t\DeltaPhisig}

\newcommand{\PhiC}{\Phi}

\usepackage{dsfont}
\usepackage{tikz-cd}
\usepackage{tikz}

\newcommand{\GraphAlignedCorner}{1.06066}
\newcommand{\GraphAlignedInner}{0.53033}
\newcommand{\GraphAlignedMid}{0.37712}

\newlength{\GraphVertexRadius}
\newlength{\GraphEdgeWidth}
\newlength{\GraphRowWidth}
\newlength{\AdditionalGraphColumnSep}
\newcommand{\GraphBoxHalf}{1.60}

\colorlet{TriangleFillColor}{blue!15}
\newcommand{\TriangleFillOpacity}{0.45}
\definecolor{GraphOuterOrange}{rgb}{0.99,0.78,0.07}
\colorlet{OctHiddenEdge}{black!50}

\newcommand{\blackvertex}[1]{\fill[black] (#1) circle[radius=\GraphVertexRadius];}

\newcommand{\orangehighlightvertex}[1]{\filldraw[fill=GraphOuterOrange, draw=black, line width=1.0pt]
(#1) circle[radius=\GraphVertexRadius];}

\newcommand{\filltriangle}[3]{\fill[TriangleFillColor, opacity=\TriangleFillOpacity]
(#1) -- (#2) -- (#3) -- cycle;}

\newcommand{\GraphBoundingBox}{\path[use as bounding box]
  (-\GraphBoxHalf,-\GraphBoxHalf)
  rectangle
(\GraphBoxHalf,\GraphBoxHalf);}

\newcommand{\GraphA}{%
  \begin{tikzpicture}[
      baseline={(0,0)},
      x=1cm,
      y=1cm,
      edge/.style={
        draw=black,
        line width=\GraphEdgeWidth,
        line cap=round,
        line join=round
      },
      every node/.style={font=\small}
    ]
    \GraphBoundingBox

    \coordinate (A) at (-\GraphAlignedCorner, \GraphAlignedCorner);
    \coordinate (B) at ( \GraphAlignedCorner, \GraphAlignedCorner);
    \coordinate (C) at ( \GraphAlignedCorner,-\GraphAlignedCorner);
    \coordinate (D) at (-\GraphAlignedCorner,-\GraphAlignedCorner);

    \coordinate (p) at (-\GraphAlignedInner, \GraphAlignedInner);
    \coordinate (q) at ( \GraphAlignedInner, \GraphAlignedInner);
    \coordinate (r) at ( \GraphAlignedInner,-\GraphAlignedInner);
    \coordinate (s) at (-\GraphAlignedInner,-\GraphAlignedInner);
    \coordinate (o) at (0,0);

    \filltriangle{A}{B}{o}
    \filltriangle{B}{C}{o}
    \filltriangle{C}{D}{o}
    \filltriangle{D}{A}{o}

    \draw[edge] (p) -- (q) -- (r) -- (s) -- cycle;
    \draw[edge] (A) -- (o) -- (C);
    \draw[edge] (D) -- (o) -- (B);
    \draw[edge] (p) -- (B);
    \draw[edge] (q) -- (C);
    \draw[edge] (r) -- (D);
    \draw[edge] (s) -- (A);

    \draw[edge] (A) -- (B) -- (C) -- (D) -- cycle;

    \foreach \v in {A,B,C,D} {
      \blackvertex{\v}
    }
    \foreach \v in {p,q,r,s,o} {
      \blackvertex{\v}
    }

    \node[above left=-1pt,  overlay] at (A) {$a$};
    \node[above right=-1pt, overlay] at (B) {$b$};
    \node[below right=-1pt, overlay] at (C) {$c$};
    \node[below left=-1pt,  overlay] at (D) {$d$};
  \end{tikzpicture}%
}

\newcommand{\GraphSubdividedTriangle}{
  \begin{tikzpicture}[
      baseline={(0,0)},
      x=0.9cm,
      y=0.9cm,
      edge/.style={
        draw=black,
        line width=\GraphEdgeWidth,
        line cap=round,
        line join=round
      },
      every node/.style={font=\small}
    ]
    \GraphBoundingBox

    \coordinate (a) at ( 0.00,  1.50);
    \coordinate (b) at (-1.45, -0.92);
    \coordinate (c) at ( 1.45, -0.92);

    \coordinate (x) at (-0.72,  0.29);
    \coordinate (y) at ( 0.00, -0.92);
    \coordinate (z) at ( 0.72,  0.29);

    \filltriangle{a}{x}{z}
    \filltriangle{b}{x}{y}
    \filltriangle{c}{y}{z}
    \filltriangle{x}{y}{z}

    \draw[edge] (a) -- (x) -- (b);
    \draw[edge] (b) -- (y) -- (c);
    \draw[edge] (c) -- (z) -- (a);
    \draw[edge] (x) -- (y) -- (z) -- cycle;

    \foreach \v in {a,b,c,x,y,z} {
      \blackvertex{\v}
    }
\end{tikzpicture}}

\newcommand{\GraphSubdividedTriangleNoYZ}{
  \begin{tikzpicture}[
      baseline={(0,0)},
      x=0.9cm,
      y=0.9cm,
      edge/.style={
        draw=black,
        line width=\GraphEdgeWidth,
        line cap=round,
        line join=round
      },
      every node/.style={font=\small}
    ]
    \GraphBoundingBox

    \coordinate (a) at ( 0.00,  1.50);
    \coordinate (b) at (-1.45, -0.92);
    \coordinate (c) at ( 1.45, -0.92);

    \coordinate (x) at (-0.72,  0.29);
    \coordinate (y) at ( 0.00, -0.92);
    \coordinate (z) at ( 0.72,  0.29);

    \filltriangle{a}{x}{z}
    \filltriangle{b}{x}{y}

    \draw[edge] (a) -- (x) -- (b);
    \draw[edge] (b) -- (y) -- (c);
    \draw[edge] (c) -- (z) -- (a);
    \draw[edge] (x) -- (y);
    \draw[edge] (x) -- (z);

    \foreach \v in {a,b,c,x,y,z} {
      \blackvertex{\v}
    }
\end{tikzpicture}}

\newcommand{\GraphOctahedronNoBottom}{
  \begin{tikzpicture}[
      baseline={(0,0)},
      x=0.9cm,
      y=0.9cm,
      edge/.style={
        draw=black,
        line width=\GraphEdgeWidth,
        line cap=round,
        line join=round
      },
      hiddenedge/.style={
        draw=OctHiddenEdge,
        line width=\GraphEdgeWidth,
        line cap=round,
        line join=round
      },
      every node/.style={font=\small}
    ]
    \GraphBoundingBox

    \coordinate (T) at ( 0.00,  1.50);
    \coordinate (L) at (-1.50, -0.45);
    \coordinate (R) at ( 1.50, -0.45);
    \coordinate (F) at (-0.36, -0.92);
    \coordinate (H) at ( 0.36,  0.02);

    \filltriangle{T}{L}{H}
    \filltriangle{T}{H}{R}
    \filltriangle{T}{L}{F}
    \filltriangle{T}{F}{R}

    \draw[hiddenedge] (T) -- (H);
    \draw[hiddenedge] (L) -- (H);
    \draw[hiddenedge] (H) -- (R);

    \draw[edge] (T) -- (L);
    \draw[edge] (T) -- (R);
    \draw[edge] (T) -- (F);
    \draw[edge] (L) -- (F);
    \draw[edge] (F) -- (R);

    \foreach \v in {T,L,R,F,H} {
      \blackvertex{\v}
    }
\end{tikzpicture}}

\newcommand{\GraphD}{%
  \begin{tikzpicture}[
      baseline={(0,0)},
      x=1cm,
      y=1cm,
      edge/.style={
        draw=black,
        line width=\GraphEdgeWidth,
        line cap=round,
        line join=round
      },
      every node/.style={font=\small}
    ]
    \GraphBoundingBox

    \coordinate (a) at (-\GraphAlignedCorner, \GraphAlignedCorner);
    \coordinate (b) at ( \GraphAlignedCorner, \GraphAlignedCorner);
    \coordinate (c) at ( \GraphAlignedCorner,-\GraphAlignedCorner);
    \coordinate (d) at (-\GraphAlignedCorner,-\GraphAlignedCorner);

    \coordinate (y) at (-\GraphAlignedMid,0);
    \coordinate (x) at ( \GraphAlignedMid,0);

    \filltriangle{a}{b}{x}
    \filltriangle{a}{x}{y}
    \filltriangle{a}{d}{y}
    \filltriangle{x}{b}{c}
    \filltriangle{x}{y}{c}
    \filltriangle{c}{d}{y}

    \draw[edge] (a) -- (b) -- (c) -- (d) -- cycle;
    \draw[edge] (y) -- (x);
    \draw[edge] (a) -- (y);
    \draw[edge] (d) -- (y);
    \draw[edge] (y) -- (c);
    \draw[edge] (a) -- (x);
    \draw[edge] (x) -- (b);
    \draw[edge] (x) -- (c);

    \foreach \v in {a,b,c,d,x,y} {
      \blackvertex{\v}
    }
  \end{tikzpicture}%
}

\newcommand{\GraphBRotated}{
  \begin{tikzpicture}[
      baseline={(0,0)},
      x=1cm,
      y=1cm,
      edge/.style={
        draw=black,
        line width=\GraphEdgeWidth,
        line cap=round,
        line join=round
      },
      every node/.style={font=\small}
    ]
    \GraphBoundingBox

    \begin{scope}[rotate=45]
      \coordinate (N) at ( 0.00, 1.55);
      \coordinate (S) at ( 0.00,-1.55);
      \coordinate (W) at (-1.55, 0.00);
      \coordinate (E) at ( 1.55, 0.00);

      \coordinate (Ndot) at ( 0.00, 1.50);
      \coordinate (Sdot) at ( 0.00,-1.50);
      \coordinate (Wdot) at (-1.50, 0.00);
      \coordinate (Edot) at ( 1.50, 0.00);

      \coordinate (O) at (0,0);

      \coordinate (NW) at (-0.9, 0.9);
      \coordinate (N1) at (-0.3, 0.9);
      \coordinate (N2) at ( 0.3, 0.9);
      \coordinate (NE) at ( 0.9, 0.9);

      \coordinate (E1) at (0.9, 0.3);
      \coordinate (E2) at (0.9,-0.3);

      \coordinate (SE) at ( 0.9,-0.9);
      \coordinate (S2) at ( 0.3,-0.9);
      \coordinate (S1) at (-0.3,-0.9);
      \coordinate (SW) at (-0.9,-0.9);

      \coordinate (W2) at (-0.9,-0.3);
      \coordinate (W1) at (-0.9, 0.3);

      \filltriangle{O}{NW}{N1}
      \filltriangle{O}{N1}{N2}
      \filltriangle{O}{N2}{NE}
      \filltriangle{O}{NE}{E1}
      \filltriangle{O}{E1}{E2}
      \filltriangle{O}{E2}{SE}
      \filltriangle{O}{SE}{S2}
      \filltriangle{O}{S2}{S1}
      \filltriangle{O}{S1}{SW}
      \filltriangle{O}{SW}{W2}
      \filltriangle{O}{W2}{W1}
      \filltriangle{O}{W1}{NW}

      \filltriangle{N}{NW}{N1}
      \filltriangle{N}{N1}{N2}
      \filltriangle{N}{N2}{NE}

      \filltriangle{S}{SW}{S1}
      \filltriangle{S}{S1}{S2}
      \filltriangle{S}{S2}{SE}

      \filltriangle{W}{NW}{W1}
      \filltriangle{W}{W1}{W2}
      \filltriangle{W}{W2}{SW}

      \filltriangle{E}{NE}{E1}
      \filltriangle{E}{E1}{E2}
      \filltriangle{E}{E2}{SE}

      \draw[edge]
      (NW) -- (N1) -- (N2) -- (NE)
      -- (E1) -- (E2) -- (SE)
      -- (S2) -- (S1) -- (SW)
      -- (W2) -- (W1) -- cycle;

      \foreach \v in {NW,N1,N2,NE} {
        \draw[edge] (N) -- (\v);
      }
      \foreach \v in {SW,S1,S2,SE} {
        \draw[edge] (S) -- (\v);
      }
      \foreach \v in {NW,W1,W2,SW} {
        \draw[edge] (W) -- (\v);
      }
      \foreach \v in {NE,E1,E2,SE} {
        \draw[edge] (E) -- (\v);
      }

      \foreach \v in {NW,N1,N2,NE,E1,E2,SE,S2,S1,SW,W2,W1} {
        \draw[edge] (O) -- (\v);
      }

      \foreach \v in {O,NW,N1,N2,NE,E1,E2,SE,S2,S1,SW,W2,W1} {
        \blackvertex{\v}
      }
      \foreach \v in {Ndot,Sdot,Wdot,Edot} {
        \blackvertex{\v}
      }
    \end{scope}
\end{tikzpicture}}

\newcommand{\GraphOctahedronFull}{%
  \begin{tikzpicture}[
      baseline={(0,0)},
      x=0.9cm,
      y=0.9cm,
      edge/.style={
        draw=black,
        line width=\GraphEdgeWidth,
        line cap=round,
        line join=round
      },
      hiddenedge/.style={
        draw=OctHiddenEdge,
        line width=\GraphEdgeWidth,
        line cap=round,
        line join=round
      },
      every node/.style={font=\small}
    ]
    \GraphBoundingBox

    \coordinate (T) at ( 0.00,  1.50);
    \coordinate (B) at ( 0.00, -0.92);
    \coordinate (L) at (-1.50,  0.29);
    \coordinate (R) at ( 1.50,  0.29);
    \coordinate (F) at (-0.36,  0.00);
    \coordinate (H) at ( 0.36,  0.58);

    \filltriangle{T}{L}{H}
    \filltriangle{T}{H}{R}
    \filltriangle{B}{H}{L}
    \filltriangle{B}{R}{H}

    \filltriangle{T}{L}{F}
    \filltriangle{T}{F}{R}
    \filltriangle{B}{L}{F}
    \filltriangle{B}{F}{R}

    \draw[hiddenedge] (T) -- (H);
    \draw[hiddenedge] (H) -- (R);
    \draw[hiddenedge] (B) -- (H);
    \draw[hiddenedge] (L) -- (H);

    \draw[edge] (T) -- (L);
    \draw[edge] (L) -- (B);
    \draw[edge] (B) -- (R);
    \draw[edge] (R) -- (T);

    \draw[edge] (T) -- (F);
    \draw[edge] (F) -- (B);
    \draw[edge] (L) -- (F);
    \draw[edge] (F) -- (R);

    \foreach \v in {T,B,L,R,F,H} {
      \blackvertex{\v}
    }
  \end{tikzpicture}%
}

\newcommand{\GraphAHighlighted}{%
  \begin{tikzpicture}[
      baseline={(0,0)},
      x=1cm,
      y=1cm,
      outeredge/.style={
        draw=black,
        line width=1.0pt,
        line cap=round,
        line join=round
      },
      grayedge/.style={
        draw=black!60,
        line width=\GraphEdgeWidth,
        line cap=round,
        line join=round
      },
      every node/.style={font=\small}
    ]
    \GraphBoundingBox

    \coordinate (A) at (-\GraphAlignedCorner, \GraphAlignedCorner);
    \coordinate (B) at ( \GraphAlignedCorner, \GraphAlignedCorner);
    \coordinate (C) at ( \GraphAlignedCorner,-\GraphAlignedCorner);
    \coordinate (D) at (-\GraphAlignedCorner,-\GraphAlignedCorner);

    \coordinate (p) at (-\GraphAlignedInner, \GraphAlignedInner);
    \coordinate (q) at ( \GraphAlignedInner, \GraphAlignedInner);
    \coordinate (r) at ( \GraphAlignedInner,-\GraphAlignedInner);
    \coordinate (s) at (-\GraphAlignedInner,-\GraphAlignedInner);
    \coordinate (o) at (0,0);

    \draw[grayedge] (p) -- (q) -- (r) -- (s) -- cycle;
    \draw[grayedge] (A) -- (o) -- (C);
    \draw[grayedge] (D) -- (o) -- (B);
    \draw[grayedge] (p) -- (B);
    \draw[grayedge] (q) -- (C);
    \draw[grayedge] (r) -- (D);
    \draw[grayedge] (s) -- (A);

    \draw[outeredge] (A) -- (B) -- (C) -- (D) -- cycle;

    \foreach \v in {A,B,C,D} {
      \orangehighlightvertex{\v}
    }
    \foreach \v in {p,q,r,s,o} {
      \blackvertex{\v}
    }

    \node[above left=-1pt,  overlay] at (A) {$a$};
    \node[above right=-1pt, overlay] at (B) {$b$};
    \node[below right=-1pt, overlay] at (C) {$c$};
    \node[below left=-1pt,  overlay] at (D) {$d$};
  \end{tikzpicture}%
}

\begin{document}

\maketitle

\begin{abstract}
  Approximate agreement tasks on graphs are discrete relaxations of consensus,
  where each process in a distributed system is given as input a vertex on a graph $G$, and processes have to output vertices that lie on a clique of $G$ contained in the convex hull of the input vertices.
  Although such tasks have been widely studied in a variety of models, graph classes and notions of convexity, it remains largely open for which classes of graphs these problems are solvable in asynchronous systems.

  In this work, we give a complete topological characterisation of the $t$-resilient solvability of approximate agreement on graphs and simplicial complexes in asynchronous shared-memory systems with read-write registers.
  As a result, we answer several open problems related to different variants of approximate agreement on graphs.
  For example, we give the first proof of Ledent's conjecture [PODC 2021] about the wait-free solvability of clique agreement.

  In fact, we show a more general result: clique agreement is $t$-resilient solvable on a graph $G$ if and only if its clique complex is $(t-1)$-connected in the homotopical sense.
  We also show that clique and monophonic agreement are solvable on the same class of graphs, but there exists a separation between monophonic and geodesic agreement, answering a question by Alistarh et al.~[TCS 2023].
  In the message-passing setting, our results imply new resilience bounds for asynchronous approximate agreement and round lower bounds for synchronous approximate agreement on graphs.
\end{abstract}

\section{Introduction}

Approximate agreement tasks are relaxations of consensus, which are often (efficiently) solvable in distributed systems in which consensus is either impossible or expensive to solve~\cite{fischer1985impossibility,pease80reaching,chor1987processor,herlihy1998unifying}.
In this work, we study the solvability of approximate agreement tasks defined on graphs and simplicial complexes in systems where processes can fail by~crashing.

\subsection{Approximate agreement tasks}

In an approximate agreement task, each process $i$  receives a local input $x_i \in V$ from some fixed set $V$ of values and needs to irrevocably decide on an output value $y_i \in V$ subject to task-specific constraints.
Instead of requiring exact agreement (i.e., that all processes output the same value) as in consensus, approximate agreement tasks require that
\begin{enumerate}[(i)]
  \item the decided outputs are \emph{close} to one another under some distance metric, and
  \item the decided outputs are contained in some problem-specific \emph{closure} of the inputs.
\end{enumerate}
Approximate agreement tasks have been studied in both continuous and discrete settings.
\vspace{-0.5em}
\subparagraph{Continuous approximate agreement.}
In the non-discrete case, the prototypical example is $\varepsilon$-approximate agreement on real values~\cite{ddolev1986reaching,abraham2005optimal,coan1988compiler,ghinea2025sok,schenk1995faster,attiya1994wait}: each process $i$ receives as input a real value $x_i$ and needs to decide on an output value $y_i$ such that (i) the distance between any two decided outputs is at most $\varepsilon>0$ and (ii) the outputs reside in the Euclidean convex hull of the  processes' inputs.
This problem naturally extends from real values to any $m$-dimensional Euclidean space $\R^m$, and the problem has subsequently been studied in various communication and fault models~\cite{mendes2015multidimensional,ghinea2026network,attiya2023step,fugger2021tight} and under different convexities~\cite{xiang2017relaxed,cambus2025centroid}.

\subparagraph{Our focus: discrete approximate agreement.}
Another line of research has investigated approximate agreement tasks on purely \emph{discrete}, combinatorial spaces, such as graphs~\cite{alcantara2019topology,nowak2019byzantine,alistarh2023wait,castaneda2018convergence,constantinescu2024convex,attiya2023multi,ledent2021brief,erbes2026asynchronous} and lattices~\cite{attiya1995atomic,faleiro2012generalized,nowak2019byzantine,di2020synchronous,zheng2020asynchronous}.
Particularly for graphs, it has remained open on which graph classes such tasks are solvable in the asynchronous setting. Indeed, this has been raised as an open problem in several papers~\cite{alcantara2019topology,nowak2019byzantine,alistarh2023wait,ghinea2025sok,ledent2021brief,liu2023impossibility,fuchs2026optimal}.

In this work, we answer this question by characterising the class of graphs which admit $t$-resilient asynchronous approximate agreement protocols.

\subparagraph{Model of computing.}
We mainly consider the standard asynchronous shared-memory model with $n \ge 1$ processes $P = \{1, \ldots, n\}$ that communicate by reading and writing to shared memory using registers with atomic read and write operations.

We say that a process \emph{terminates} if it performs finitely many read and write operations.
We say that a process \emph{crashes} if it terminates without deciding on an output value.
For $0 \le t < n$, a protocol is \emph{$t$-resilient} if in any execution with at most $t$ crashes all non-crashed processes decide on an output value. A protocol is \emph{wait-free} if it is $(n-1)$-resilient.

\subsection{Approximate agreement on graphs}

Let $G = (V,E)$ be a graph.
There are three commonly studied variants of approximate agreement on $G$: clique, monophonic, and geodesic agreement.
In each variant,
every  process $i \in P$ gets as input a vertex $x_i \in V$ and must irrevocably decide an output vertex $y_i \in V$.
If a process $i$ crashes before deciding its output, we write $y_i = \bot$.
Let $X = \{ x_i : i \in P \}$ be the set of inputs given to the processes and
$Y = \{y_i : y_i \neq \bot, i \in P\}$ be the set of outputs.

Typical variants of approximate agreement on a graph $G$ relax the  \emph{agreement condition} of consensus as follows:
The set $Y$ of outputs must be a clique of $G$.
That is, the outputs of non-faulty processes have to be within distance 1 under the natural distance metric on $G$.

\subparagraph{Different validity conditions.}
For the \emph{validity condition}, we consider three variants, all of which lead to different tasks with increasingly stronger output guarantees:
\begin{itemize}
  \item In  \emph{clique agreement},  validity requires that if the inputs $X$ form a clique of $G$, then $Y \subseteq X$.
  \item In \emph{monophonic agreement}, the validity condition requires that any output value lies on an \emph{induced path} between some input vertices in $X$.
    Put otherwise, the set $Y$ of outputs must be contained in the monophonic (also known as the minimal path) convex hull of $X$.
  \item In \emph{geodesic agreement}, the validity condition requires that any output value lies on a \emph{shortest path} between some input vertices in $X$.
    Put otherwise, the set $Y$ of outputs must be contained in the geodesic (also known as the shortest path) convex hull of $X$.
\end{itemize}
Clearly, clique agreement is the weakest task, as the other two variants also solve clique agreement.
Furthermore, geodesic agreement is at least as hard as monophonic agreement,
because any shortest path is an induced path.
In distance-hereditary graphs, such as trees and block graphs, geodesic and monophonic agreement are the same task, because in these graphs all induced paths are also shortest paths. However, it has remained open whether these tasks are always solvable on the same class of graphs.

\subparagraph{Prior work on approximate agreement on graphs.}
In the wait-free setting,
clique agreement was first studied in trees by Casta{\~n}eda, Rajsbaum, and Roy~\cite{castaneda2018convergence} and on general graphs by Alcántara et al.~\cite{alcantara2019topology} under the name 1-gathering. The monophonic and geodesic variants were introduced by Nowak and Rybicki~\cite{nowak2019byzantine} in the Byzantine message-passing setting.
When $G$ is a path, geodesic and monophonic approximate agreement correspond to the natural discretisation of approximate agreement on real values~\cite{coan1988compiler,abraham2005optimal,fekete1990asymptotically,fekete1994asynchronous,ddolev1986reaching,schenk1995faster}.

In the case of general graphs, these approximate agreement tasks have so far been studied on several different graph classes, such as trees~\cite{castaneda2018convergence,nowak2019byzantine,erbes2026asynchronous,fuchs2026optimal}, graphs that have a clique tree~\cite{alcantara2019topology}, chordal graphs~\cite{nowak2019byzantine,constantinescu2024convex}, nicely bridged graphs~\cite{alistarh2023wait} and triangulations of spheres~\cite{liu2023impossibility}.
In the Byzantine message-passing setting, monophonic agreement has asynchronous protocols for chordal graphs~\cite{nowak2019byzantine,constantinescu2024convex} whenever $n > (\omega+1)t$, where $\omega$ is the size of the largest clique.
Recently, Fuchs et al.~\cite{fuchs2026optimal} obtained round- and resilience-optimal synchronous protocols on trees and block graphs whenever $n > 3t$.
A variant on trees, known as connected consensus, corresponds to multivalued generalisations of many fundamental communication~primitives~\cite{arteaga2026time,attiya2023multi}.

The \emph{exact} variants of these problems correspond to the convex consensus task on the given convexity space. This problem is now well-understood with asymptotically optimal round complexity~\cite{nowak2019byzantine} and optimal resilience thresholds~\cite{constantinescu2024convex} in the synchronous setting.

\subparagraph{The open problem: How hard is approximate agreement on graphs?}
Despite substantial progress in recent years on approximate agreement on graphs,
the class of graphs in which \emph{any} of the above tasks is solvable has remained elusive.
Indeed, it is not even known on which graph classes these problems are wait-free solvable~\cite{ledent2021brief,alistarh2023wait,liu2023impossibility}, let alone $t$-resilient~solvable.

Furthermore, even seemingly simpler questions, such as whether monophonic or geodesic agreement is strictly harder than clique agreement, have been open. Indeed,
all previously known impossibility results hold for the weakest variant, i.e., clique agreement~\cite{alistarh2023wait,alcantara2019topology,liu2023impossibility,fuchs2026optimal},
whereas the state-of-the-art wait-free protocols~\cite{alistarh2023wait} work even under the strongest geodesic validity condition in nicely bridged graphs (which include all chordal graphs).

\subsection{Simplex agreement and Ledent's conjecture}

Ledent~\cite{ledent2021brief} identified an elegant connection between approximate agreement on graphs and approximate agreement on higher-dimensional objects called \emph{simplicial complexes}.
An \emph{abstract simplicial complex} is a finite collection $\A$ of non-empty sets that is downward closed: if $\sigma \in \A$, then $\tau \in \A$ for any non-empty $\tau \subseteq \sigma$.
The sets in $\A$ are called the \emph{simplices} of the complex $\A$ and the elements of $V(\A) = \bigcup_{\sigma \in \A} \sigma$ are called the \emph{vertices} of the complex.

\subparagraph{Simplex agreement.}
Ledent~\cite{ledent2021brief} defined the \emph{simplex agreement} task on an abstract simplicial complex $\A$ as follows. Each process $j$ receives as input a vertex $x_j \in V(\A)$ and must output a vertex $y_j \in V(\A)$. As before, let $X$ be the set of vertices given as inputs to processes and $Y$ the set of outputs decided by the processes.
Simplex agreement has the following constraints:
\begin{enumerate}[(i)]
  \item Agreement: The decided outputs form a simplex of $\A$. That is, $Y \in \A$.

  \item Validity:  If the inputs form a simplex, then the decided outputs must lie on that simplex. That is,  if $X \in \A$, then $Y \subseteq X$.
\end{enumerate}
We note that in the literature ``simplex agreement'' has sometimes been used for other tasks, such as barycentric agreement~\cite{herlihy1993asynchronous}, variants of loop agreement~\cite{herlihy2003classification}, and trivial tasks~\cite{liu2009classifying}.
All these are different tasks; in this work, we use simplex agreement in the sense of~Ledent~\cite{ledent2021brief}.

\begin{figure}[t]
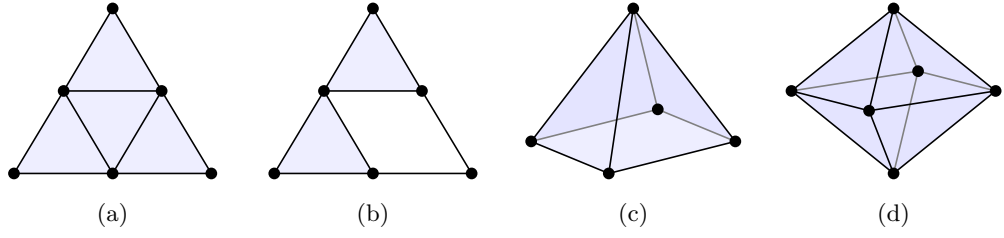

  \centering

  \makebox[\textwidth][c]{
    \begin{tabular}{@{}c@{\hspace{\AdditionalGraphColumnSep}}c@{\hspace{\AdditionalGraphColumnSep}}c@{\hspace{\AdditionalGraphColumnSep}}c@{}}
      \shortstack{\GraphSubdividedTriangle\\[-1em](a)}
      &
      \shortstack{\GraphSubdividedTriangleNoYZ\\[-1em](b)}
      &
      \shortstack{\GraphOctahedronNoBottom\\[-1em](c)}
      &
      \shortstack{\GraphOctahedronFull\\[-1em](d)}
  \end{tabular}}
  \caption{Graphs and their clique complexes. The black points represent vertices and the lines represent edges of the graph. The shaded triangles correspond to cliques of size three, which are simplices of the clique complex.
  (a) A chordal graph with a contractible clique complex. (b) A graph whose clique complex is not contractible. (c) A non-chordal graph with contractible clique complex. (d) The octahedron graph. The clique complex is 1-connected, but not contractible.}
  \label{fig:intro-graphs}
\end{figure}

\subparagraph{Ledent's conjecture.}
Ledent~\cite{ledent2021brief} observed that \emph{clique agreement} on any graph $G$ is equivalent to simplex agreement on its clique complex $\K(G)$, which is the simplicial complex
\[
  \K(G) = \{ \sigma \subseteq V(G) : \sigma \text{ is a clique of } G \}.
\]
He further conjectured the following interesting connection between the wait-free solvability of clique agreement on $G$ and the topological properties of its clique complex:
\begin{quote}
  \emph{Ledent's conjecture: }
  Clique agreement is wait-free solvable on a graph $G$ for any number of processes if and only if its clique complex $\K(G)$  is contractible.
\end{quote}
Contractibility is a topological notion, which does not have a purely combinatorial definition.
Informally, a simplicial complex is contractible if its \emph{geometric realisation}, i.e., its natural embedding in Euclidean space, can be continuously deformed into a single point.
\Cref{fig:intro-graphs} gives examples of graphs and their clique complexes.
The graphs in \Cref{fig:intro-graphs}a and \Cref{fig:intro-graphs}c have contractible clique complexes, whereas the   graphs in \Cref{fig:intro-graphs}b and \Cref{fig:intro-graphs}d do not.

\subparagraph{Prior work on Ledent's conjecture.}
Some supporting evidence supporting Ledent's conjecture previously existed. Alistarh et al.~\cite{alistarh2023wait} showed that there exist no wait-free protocols for $n \ge 3$ processes on graphs which admit a certain ``impossibility labelling''; \Cref{fig:intro-graphs}b gives an example of such a graph. Ledent noted that graphs admitting such labellings have non-simply connected clique complexes -- and hence, they are also non-contractible.
On the other hand, chordal graphs and bridged graphs have contractible clique complexes, and these graphs are known to have wait-free algorithms~\cite{alistarh2023wait}.

Ledent~\cite{ledent2021brief} identified the octahedron graph (\Cref{fig:intro-graphs}d) as an example that is (i) not covered by the impossibility result of Alistarh et al.~\cite{alistarh2023wait}, but (ii) has a non-contractible clique complex, because the graph can be obtained as a triangulation of the two-dimensional sphere. Following this, Liu~\cite{liu2023impossibility} gave further evidence in support of the conjecture by showing that indeed there are no wait-free clique agreement protocols for $n\ge 4$ processes on graphs which can be obtained as triangulations and stellations of $d$-dimensional spheres for $d \ge 2$. Liu also gave a new condition for graphs which do not have wait-free protocols, but did not provide a complete characterisation of the class of wait-free solvable graphs.

\subsection{Summary of our contributions}

In this work, we give a complete topological characterisation of the wait-free and $t$-resilient solvability of simplex agreement on simplicial complexes and the solvability of clique agreement, monophonic and geodesic agreement on graphs.
Thus, our work settles several open problems about approximate agreement on graphs raised in the literature~\cite{alistarh2023wait,nowak2019byzantine,ledent2021brief,liu2023impossibility,alcantara2019topology}.

In particular, we give a proof of Ledent's conjecture about solvability of clique agreement~\cite{ledent2021brief}.
Our proof covers not only the wait-free setting, but also the more general $t$-resilient case, and applies also to the more stringent monophonic and geodesic validity conditions rather than just the weakest clique agreement version.
We further show that clique agreement and monophonic agreement are solvable on the same class of graphs.
However, we show a separation between geodesic agreement and monophonic agreement.

Our results also imply resilience bounds for the asynchronous message-passing model and round complexity lower bounds for the synchronous message-passing model, improving the state-of-the-art bounds for large classes of graphs in the message-passing setting as well~\cite{alistarh2023wait,fuchs2026optimal}.
Finally, we show that it is undecidable to determine if a given graph $G$ has a $t$-resilient protocol for clique, monophonic or geodesic agreement on $G$.
We give a more detailed overview of our results in \Cref{sec:overview} after establishing the necessary technical background.

\section{Preliminaries}

\subparagraph{Graphs.}
Let $G = (V,E)$ be a graph with vertex set $V(G) = V$ and edge set $E(G) = E$.
For any $U \subseteq V$, we write $G[U]$ for the subgraph induced by $U$.
A path $\rho =  \{v_0, \ldots, v_k\}$ is \emph{induced} if $G[\rho]$ is a path.
The open neighbourhood of a vertex $v$ in $G$ is $N_G(v)= \{ u : \{u,v\} \in E \}$ and its closed neighbourhood is $N_G[v] = N_G(v) \cup \{v\}$.

A vertex $v$ is said to be \emph{dominated} in $G$ if there is some $u \neq v$  such that $N_G[v] \subseteq N_G[u]$.
The graph $G$ is \emph{dismantlable} if we can order the vertices $v_1, \ldots, v_k$ such that for any $1\le i < k$ the vertex $v_i$ is dominated in the subgraph induced by $\{v_{i}, v_{i+1}, \ldots, v_k\}$.

\subparagraph{Convexity spaces.}
A \emph{convexity space} on a finite set of values $V$ is a collection $\C$ of subsets of $V$ such that (i) $\emptyset, V \in \C$ and (ii) $\C$ is closed under intersections.
For $U \subseteq V$, the \emph{convex hull} $\langle U \rangle_\C$ of $U \subseteq V$ is the smallest set $S \in \C$ such that $U \subseteq S$.
Note that the convex hull operator is a closure operator.

There are many natural notions of convexity~\cite{farber1987local} that can be defined on a graph $G$. Let $U \subseteq V(G)$.
The set $U$ is \emph{monophonically convex} in $G$ if it contains all the vertices on all \emph{induced} paths between any pair of vertices in $U$.
The set $U \subseteq V(G)$ is \emph{geodesically convex} if it contains all the vertices on all \emph{shortest} paths between any pair of vertices in $U$.
We can also define the notion of clique convexity as follows: a set $U \subseteq V(G)$ is \emph{clique convex} if $U$ is a clique of $G$, $U=\emptyset$, or $U = V(G)$.
Observe that the validity property of monophonic, geodesic, and clique agreement uses the respective convex hull as the closure operation.

\subparagraph{Abstract simplicial complexes.}
Recall that an \emph{abstract simplicial complex} $\A$ is a finite collection of non-empty sets such that
if $\sigma \in \A$, then $\tau \in \A$ for all $\emptyset \neq \tau \subseteq \sigma$.
An element $\sigma \in \A$ is called a \emph{simplex} of $\A$.
The elements of the set $V(\A) = \bigcup_{\sigma \in \A} \sigma$ are called the \emph{vertices} of $\A$.
For two abstract simplicial complexes $\A$ and $\B$, we say that $\B$ is a \emph{subcomplex} of $\A$ if
$\B \subseteq \A$. For $U \subseteq V( \A)$, we write $\A[U]$ for the subcomplex $\{ \sigma \in \A : \sigma \subseteq U\}$ of $\A$ induced by $U$.

If $\sigma$ is a simplex consisting of $k+1$ vertices, then the \emph{dimension} of $\sigma$ is $\dim(\sigma) := k$. Given a finite set $V$ of size $k+1$,  we define the (abstract) \emph{standard $k$-simplex} on the
vertices $V$ as the simplicial complex
\[
  \Delta^V := 2^V \setminus \{\emptyset\} = \Set{\tau \subseteq V}{\tau \neq \emptyset}.
\]
That is, $\Delta^V$ is the collection of all non-empty subsets of $V$.
When the vertex set is $\{0,1,\ldots,k\}$, we write $\Delta^k$ instead of $\Delta^{\{0,1,\ldots,k\}}$. Moreover, given a simplex $\sigma \in \A$ we denote by $\Delta^\sigma \subseteq \A$ the subcomplex of $\A$
that is induced by $\sigma$. The \emph{$k$-skeleton} of a simplicial complex $\A$ is the subcomplex consisting of all simplices of dimension at most $k$, that is,
\begin{equation*}
  \skel^k \A := \Set{\sigma \in \A}{\dim(\sigma) \leq k}.
\end{equation*}
The $1$-skeleton of a complex $\A$ corresponds to a graph, called the \emph{underlying graph} $G$ of $\A$, with its 1-simplices being called edges. That is, $V(G) = V(\A)$ and $E(G) =
\Set{\sigma \in \A}{\dim(\sigma) = 1}$.
The \emph{boundary} $\partial\sigma$ of a simplex $\sigma$ of dimension $k$ is the simplicial complex
\[
  \partial\sigma := \skel^{k-1}\sigma :=\Delta^\sigma\setminus\set{\sigma}.
\]

\subparagraph{Geometric realisation of a simplicial complex.}
Let $\A$ be an abstract complex with vertices $V(\A) = \set{v_0, \ldots, v_d}$. Choose any bijection $v_i \mapsto e_i$ from the vertices of $\A$ to the $d+1$ standard coordinate
basis vectors $e_0, \ldots, e_d \in \R^{d+1}$. We define the \emph{geometric realisation}~of~$\A$~as
\begin{equation*}
  |\A| := \bigcup_{\sigma \in \A} \conv \Set{e_i}{v_i \in \sigma} \subseteq \R^{d+1},
\end{equation*}
where $\conv V$ denotes the Euclidean convex hull of a finite set $V \subseteq \R^{d+1}$ of vertices.
The geometric realisation of a simplex $\sigma \in \A$ is
$|\sigma| := \conv \Set{e_i}{v_i \in \sigma}$.
Note that $|\A|$ is a topological space, as it is a subspace of $\R^{d+1}$.
Moreover,
the choice of bijection does not affect the topological properties of $|\A|$, as different bijections result in homeomorphic spaces.

\subparagraph{Carrier maps.}
A \emph{carrier map} $\Phi \colon \A \to 2^\B$ maps each simplex $\sigma \in \A$ to a subcomplex $\Phi(\sigma) \subseteq \B$ such that
\(
  \Phi(\tau) \subseteq \Phi(\sigma)
\)
for any $\tau \subseteq \sigma \in \A$.
Given a carrier map $\Phi \colon \A \rightarrow 2^\B$ and a continuous map $f \colon |\A| \rightarrow |\B|$, we say that
$f$ is carried by $\Phi$ if $f(|\sigma|) \subseteq |\Phi(\sigma)|$ for all $\sigma \in \A$.
\subparagraph{Colourless tasks.}
A \emph{colourless task} is a tuple $(\I, \Ocal, \Phi)$, where
$\I$ is a simplicial \emph{input complex},
$\Ocal$ is a simplicial \emph{output complex}, and
$\Phi \colon \I \rightarrow 2^\Ocal$ is a carrier map.
The input complex $\I$ describes the valid input assignments, the output complex $\Ocal$ is the collection of all feasible outputs, and the carrier map $\Phi$ maps valid input assignments to feasible outputs:
If $X = \{x_i : i  \in P\} \in \I$ is the set of inputs given to the processes and $Y = \{y_i  : y_i \neq \bot, i \in P\}$ is the set of decided outputs, then the outputs are a solution to the task if and only if~$Y \in \Phi(X)$.

Simplex agreement on a simplicial complex $\A$ and the different variants of approximate agreement on a graph $G$ are examples of colourless tasks.

\subparagraph{The asynchronous computability theorem.}
As we consider only solvability, we can formally work in the
\emph{layered snapshot model},
because any task is solvable in this model if and only if it is solvable in the model with registers~\cite[Chapter 5]{herlihy2013distributed}.
We skip the precise definition of this model, as we only make use of it via the $t$-resilient variant of the celebrated asynchronous computability theorem~\cite{herlihy1999topological,herlihy1993asynchronous}; see, e.g.,~\cite{herlihy2013distributed} for a modern treatment of the result.

\begin{theorem}\label{thm:t-res-computability}
  Let $0 \le t < n$. The colourless task $(\mathcal{I}, \mathcal{O}, \Phi)$ has a $t$-resilient $n$-process protocol in the layered snapshot model if and only if there exists a continuous map
  \(    f \colon |\skel^{t} \mathcal{I}| \rightarrow |\mathcal{O}| \)
  carried by $\Phi$.
\end{theorem}

Because snapshot objects can be simulated in asynchronous message-passing~\cite{attiya1995sharing}, there exists a similar result for the standard asynchronous message-passing model; again see~\cite{herlihy2013distributed}.

\begin{theorem}\label{thm:t-res-computability-msg}
  Let $0 \le t < n/2$. The colourless task $(\mathcal{I}, \mathcal{O}, \Phi)$ has a $t$-resilient $n$-process asynchronous message-passing protocol if and only if there exists a continuous map
  \(  f \colon |\skel^{t} \mathcal{I}| \rightarrow |\mathcal{O}| \)
  carried by $\Phi$.
\end{theorem}

We can lift impossibility results obtained using the asynchronous computability theorem to get synchronous message-passing lower bounds using Gafni's simulation result~\cite{gafni1998round}.

\begin{theorem}
  \label{thm:round-by-round}
  Let $0 < t < f < n$ such that $t + f < n$. If there exists a synchronous $f$-resilient $n$-process message-passing protocol that solves a colourless task
  in at most $f/t$ rounds, then the task has a $t$-resilient asynchronous shared-memory protocol.
\end{theorem}

\subparagraph{Homotopy equivalence.}
For our main result, we recall some basic concepts from homotopy theory; see, e.g.,~\cite{Hatcher} for an introduction to algebraic topology.
Let $A,B$ be topological spaces and let \mbox{$f, g \colon A \rightarrow B$} be two continuous maps.
The maps $f$ and $g$ are \emph{homotopic}, written as $f \he g$, if there exists a continuous map $H \colon [0,1] \times A  \rightarrow B$ such that
\begin{enumerate}
  \item $H(0,x) = f(x)$ for all $x \in A$, and
  \item $H(1,x) = g(x)$ for all $x \in A$.
\end{enumerate}
Such a map $H$ is called a \emph{homotopy}.
Two topological spaces $A, B$ are homotopy equivalent, written as $A \he B$, if there exist continuous maps $f \colon A \rightarrow B$ and $g: B \rightarrow A$ such that their
compositions $f \circ g$ and $g \circ f$ are respectively homotopic to the identity maps $\id_B$ and $\id_A$. In this case we call $f$ a \emph{homotopy equivalence}. The relation
$\he$ is an equivalence relation for both maps and spaces.

\subparagraph{Homotopy groups and induced maps.}
For a topological space $A$, we use $\pi_m(A)$ to denote its $m$-th \emph{homotopy group} (we refer to \cite[Section 4.1]{Hatcher} for a thorough introduction of homotopy groups).
While $\pi_m(A)$ is always a group for $m>0$, technically, $\pi_0(A)$ is not a group, but rather a set.
We say that $\pi_m(A)$ is \emph{trivial} if it consists of a single element.
In the case of $m>0$, this means that $\pi_m(A)$ is the trivial group.
When $\pi_m(A)$ is trivial, we write $\pi_m(A)=0$.
Geometrically, the $m$th homotopy group measures $m$-dimensional holes in the space $A$: if it is trivial,
there are no ``$m$-dimensional holes'' in $A$.

One of the most important properties of homotopy groups is their functoriality:
Any continuous map $f\colon A \rightarrow B$ induces for each $m \geq 0$ a map of homotopy groups
$f_\ast^{(m)}\colon \pi_m(A) \to \pi_m(B)$ which is a group homomorphism for each $m \geq 1$.
For $m=0$, we will still use the term \emph{isomorphism} rather than \emph{bijection} for the sake of simplicity in the exposition. The induced map
satisfies the following properties for all $m \geq 0$ (see, e.g.,~\cite[Section 4.1]{Hatcher}):
\begin{enumerate}
  \item If $f\colon A \to B$ and $g\colon B \to C$ are continuous maps, then $(g \circ f)_\ast^{(m)} = g_\ast^{(m)} \circ f_\ast^{(m)}$.
  \item If $f,g\colon A \to B$ are continuous maps such that $f \simeq g$, then $f_\ast^{(m)} = g_\ast^{(m)}$.
  \item If $f \colon A \to B$ is a homotopy equivalence then $f_\ast^{(m)}$ is an isomorphism for all $m \geq 0$.
\end{enumerate}

\subparagraph{Homotopical connectivity.}
Fix a base point $x_0 \in A$. A topological space $A$ is \emph{contractible} if the identity map $\id_A$ is homotopic to the constant map \mbox{$\Tilde{x}_0 \colon A \rightarrow \set{x_0}$}, that is, there exists a homotopy
\( H\colon [0,1] \times A \rightarrow A \)
such that $H(0, x) = x$ and $H(1, x) = x_0$ for all $x \in A$. The topological space $A$ is called \emph{$m$-connected} if its first $m+1$ homotopy groups are trivial.
For notational convenience, we define that any topological space is always $(-1)$-connected.
We say that an abstract simplicial complex is $m$-connected if its geometric realisation $|\A|$ is $m$-connected.

It is a standard fact in topology that contractible spaces are always $m$-connected for
every $m \geq 0$. The converse, however, is a special case of the Whitehead Theorem; see \cite[Theorem 4.5]{Hatcher}.
It does not hold for all topological spaces, but rather for a restricted class of spaces, which includes  the class of simplicial complexes.

\begin{theorem} \label{Whitehead}
  Let $\A$ be a simplicial complex which is $m$-connected for every $m \geq 0$. Then $\A$ is contractible.
\end{theorem}

\section{Overview of our results}\label{sec:overview}

\subsection{Solvability of convex simplex agreement tasks}

In this work, we define a new notion of \emph{convex simplex agreement} on simplicial complexes. With this task, we can easily generalise Ledent's observation about the equivalence of clique agreement on $G$ and simplex agreement on the clique complex $\K(G)$ to cover a wider variety of approximate agreement tasks on graphs, such as monophonic and geodesic agreement.

\begin{definition}\label{def:convex-agreement}
  Let $\A$ be a simplicial complex and $\C$ be a convexity space on $V(\A)$.
  The $\C$-convex simplex agreement task is the colourless task $(\Deltava, \A, \PhiC)$, where
  \[
    \PhiC(\sigma) :=
    \begin{cases}
      \Delta^\sigma & \text{if } \sigma  \in \A \\
      \A[\langle \sigma \rangle_\C] & \text{otherwise.}
    \end{cases}
  \]
\end{definition}
In words, the $\C$-convex simplex agreement task on $\A$ has the following constraints:
\begin{enumerate}[(i)]
  \item The outputs $Y$ form a simplex of $\A$. That is, $Y \in \A$.
  \item If the inputs $X$ satisfy $X \in \A$, then $Y \subseteq X$.
  \item The outputs are in the $\C$-convex hull of the inputs. That is, $Y \subseteq \langle X  \rangle_\C$.
\end{enumerate}
With this definition, clique, monophonic, and geodesic agreement on $G$ are now instances of convex simplex agreement on the clique complex $\K(G)$ with the appropriate convexity space.

Our main technical tool is a topological result that can be used to characterise the solvability of convex simplex agreement tasks.
In \Cref{sec:main-tool}, we prove the following:

\begin{restatable}{theorem}{maintool}\label{thm:main-tool}
  Let $(\Deltava, \A, \PhiC)$ be the $\C$-convex simplex agreement task on $\A$ and let $t \geq 0$.
  Then the following are~equivalent:
  \begin{enumerate}[(a)]
    \item There exists a continuous map
      \(
        f \colon |\skel^t \Deltava| \rightarrow |\A|
      \)
      carried by $\PhiC$.

    \item The complex $\PhiC(\sigma) \subseteq \A$ is $(t-1)$-connected for each $\sigma  \in \Deltava$.
  \end{enumerate}
\end{restatable}

Note that the above result is purely topological; it does not depend on any particular model of computing.
However, the result has several interesting consequences, as we can now
invoke the asynchronous computability theorem and its variants in different models.

At first, it may be tempting to simplify the condition in (b) as follows: $\A[S]$ is $(t-1)$-connected for each non-empty convex set $S \in \C$.
However, this simplification only works for convexity spaces, where $\A \subseteq \C$.
This is true for clique, monophonic and geodesic convexity of the underlying graph of $\A$, but in general,
simplices need not necessarily be convex sets.

\subsection{Solvability of asynchronous simplex and clique agreement}

We first discuss implications for the asynchronous shared-memory model with single-writer multi-reader registers; we consider message-passing models later.
The first immediate consequence of \Cref{thm:main-tool} and the $t$-resilient asynchronous computability theorem (\Cref{thm:t-res-computability}) is a characterisation of $t$-resilient solvability of convex simplex agreement.

\begin{corollary}\label{coro:convex-agreement-solvability}
  Let $ 0 \le t < n$. The $\C$-convex simplex agreement task on $\A$ has a $t$-resilient $n$-process protocol if and only if  $\PhiC(\sigma)$ is $(t-1)$-connected for each $\sigma \in \Deltava$.
  Moreover,
  if $\A \subseteq \C$, then the task has a $t$-resilient protocol  if and only if  $\A[S]$ is $(t-1)$-connected for each non-empty $S \in \C$.
\end{corollary}
\begin{proof}
  By \Cref{thm:t-res-computability}, the colourless task $(\Deltava, \A, \PhiC)$ is solvable if and only if there exists a continuous map $f \colon |\skel^t \Deltava | \to |\A|$ carried by $\PhiC$.
  By \Cref{thm:main-tool} such a map $f$ exists if and only if
  $\PhiC(\sigma) \subseteq \A$ is $(t-1)$-connected for each  $\sigma \in  \Deltava$.
  The second claim follows by observing that $\A \subseteq \C$ implies that
  $\PhiC(\sigma) = \A[\langle \sigma \rangle_\C]$ for any $\sigma \in \Deltava$.
\end{proof}

Observe that Ledent's simplex agreement on $\A$ is $\C$-convex simplex agreement in the convexity space given by $\C = \A \cup \{\emptyset, V(\A)\}$.
As all simplices are trivially contractible, only the contractibility of $\A$ matters.
Thus, \Cref{coro:convex-agreement-solvability} gives us the following characterisation:

\begin{restatable}{theorem}{sathm}\label{thm: SAthm}
  Let $0 \le t < n$. Simplex agreement on a simplicial complex~$\A$ has a $t$-resilient $n$-process shared-memory protocol if and only if $\A$ is $(t-1)$-connected. Moreover, simplex agreement on $\A$ is
  wait-free solvable for any $n \ge 1$
  if and only if $\A$ is contractible.
\end{restatable}

As a direct corollary, we get a proof of a generalised version of Ledent's conjecture~\cite{ledent2021brief} and a topological characterisation of the class of graphs on which clique agreement is solvable.

\begin{restatable}{corollary}{cacorollary}\label{coro:ca}
  Let $0 \le t < n$.
  Clique agreement on a graph $G$ has a $t$-resilient $n$-process shared-memory protocol if and only if its clique complex $\K(G)$ is $(t-1)$-connected. Moreover, clique agreement on $G$ is
  wait-free solvable for any $n \ge 1$
  if and only if $\K(G)$~is~contractible.
\end{restatable}

We note that \Cref{coro:ca} also answers an open problem posed by Liu~\cite{liu2023impossibility}.
Liu showed that clique agreement is not wait-free solvable on the octahedron graph (\Cref{fig:intro-graphs}d) with $n \ge 4$ processes, but asked if it remains wait-free solvable for $n=3$ processes.
The answer is yes: the octahedron graph is a triangulation of the 2-dimensional sphere; thus, its clique complex is homeomorphic to the 2-dimensional sphere, which is 1-connected.

\subsection{Solvability of asynchronous monophonic agreement on graphs}

Observe that while clique agreement on $G$ is trivially equivalent to simplex agreement on $\K(G)$, the same is not true for monophonic agreement on $G$. This is because a solution to clique agreement does not necessarily lie in the monophonic convex hull of the inputs.

However, we prove a result that characterises the homotopical connectivity of the subcomplexes of $\K(G)$ induced by monophonic convex hulls of vertex sets in the graph~$G$.

\begin{restatable}{theorem}{monothm}\label{thm:monophonic-sets-are-connected}
  Let $G$ be a graph and let $S \subseteq V(G)$ be a non-empty monophonically convex set in $G$. If the complex $\K(G)$ is $t$-connected for some $t \geq 0$, then $\K(G[S])$ is~$t$-connected.
\end{restatable}

As an immediate consequence of
\Cref{coro:convex-agreement-solvability} and \Cref{thm:monophonic-sets-are-connected}, we get that clique agreement and monophonic agreement tasks are solvable in exactly the same class of graphs.

\begin{restatable}{corollary}{cliqueminimalequiv}\label{coro:clique-minimal-paths-equivalence}
  Let $G$ be a graph and $0 \le t < n$. Clique agreement on $G$ has a $t$-resilient  $n$-process protocol if and only if monophonic agreement on $G$ has a $t$-resilient $n$-process~protocol.
\end{restatable}

\subsection{Separation between monophonic and geodesic agreement}

\begin{figure}[t]
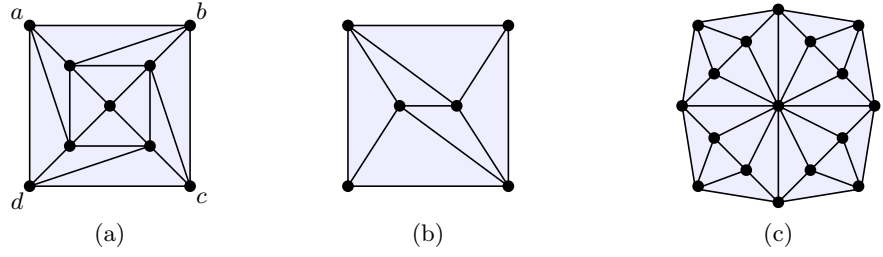

  \centering

  \makebox[\textwidth][c]{
    \begin{tabular*}{\GraphRowWidth}{@{}c@{\extracolsep{\fill}}cc@{}}
      \shortstack{\GraphA\\[-0.5em](a)}
      &
      \shortstack{\GraphD\\[-0.5em](b)}
      &
      \shortstack{\GraphBRotated\\[-0.5em](c)}
  \end{tabular*}}
  \caption{Interesting graphs for geodesic agreement. (a) A graph with a contractible clique complex that does not have a wait-free geodesic agreement
    protocol for $n=3$ processes. (b) A dismantlable graph that is neither bridged nor radius-1.
    Geodesic agreement is wait-free solvable for $n \ge 1$.
  (c) A non-dismantlable graph that has a wait-free geodesic agreement protocol for $n \ge 1$.}
  \label{fig:graphs}
\end{figure}

The state-of-the-art wait-free protocols for nicely bridged and chordal graphs work under the stronger geodesic validity and not just monophonic validity~\cite{alistarh2023wait}.
Indeed, Alistarh et al.~\cite{alistarh2023wait} asked if geodesic and monophonic
agreement are solvable on the same class of graphs.

Given the situation with clique agreement and monophonic agreement, one might now hope that the geodesic variant is also wait-free solvable on $G$ whenever $\K(G)$ is contractible. However, this turns out not to be true:
using \Cref{thm:main-tool}, we can show that geodesic agreement is strictly harder than monophonic agreement.

\begin{corollary}
  There exists a graph $G$ on which clique agreement admits a wait-free 3-process protocol, whereas geodesic agreement does not.
\end{corollary}

\Cref{fig:graphs}a gives a graph that separates monophonic and geodesic agreement. The graph has a contractible clique complex, but the geodesic convex hull of the vertices $\{a,b,c\}$ induces a 4-cycle whose clique complex is homeomorphic to the non-contractible 1-dimensional~sphere.

However, there is also good news: we show that geodesic agreement is wait-free solvable on any  \emph{dismantlable graph}.
Dismantlable graphs are also known as copwin graphs~\cite{nowakowski1983vertex}.
This is a strictly larger class of graphs than on which the problem was previously known to be wait-free solvable, as it is a strict superset of both bridged and radius-1 graphs~\cite{alistarh2023wait}.
\Cref{fig:graphs}b gives an example of a graph that is dismantlable, but is neither bridged nor radius-1.

\begin{restatable}{corollary}{corodis}
  If $G$ is a dismantlable graph, then geodesic agreement on $G$ is wait-free solvable for any number of processes.
\end{restatable}

However, we note that there exist non-dismantlable graphs for which geodesic agreement is wait-free solvable for $n \ge 1$ processes; \Cref{fig:graphs}c gives an example. Hence, dismantlability is not a complete characterisation of wait-free solvability for geodesic agreement.

\subsection{Undecidability results}

Now that we have topological characterisations for solvability of various  approximate agreement tasks on graphs,
one might ask for simpler, natural ``combinatorial'' characterisations without the need for invoking topological notions of contractibility or homotopy theory.
Unfortunately, we show that no decidable characterisations exist for any of the problems.

It is well-known that the existence of wait-free protocols for general tasks is undecidable~\cite{herlihy1997decidability,gafni1998three}. However, one can ask if the existence of $t$-resilient protocols is decidable for more restricted tasks. Alcántara et al.~\cite{alcantara2019topology} showed that the existence of wait-free clique agreement protocols is undecidable for $n=3$.
We give more bad news: it is undecidable whether a given simplicial complex is $(t-1)$-connected for any $t \ge 2$, so \Cref{thm: SAthm} implies that it is undecidable if simplex agreement has a $t$-resilient protocol for any $t \ge 2$ and $n > t$.

One may still hope that the solvability of any of the more restricted clique, monophonic, and geodesic agreement tasks on \emph{graphs} is decidable.
For clique and monophonic agreement, this boils down to asking if the connectivity problem becomes decidable when restricted to
\emph{flag complexes}, that is, simplicial complexes arising as a clique complex of some graph. Unfortunately, this restriction turns out not to be significant: up to homotopy equivalence, the class of flag complexes is just as large as the class of all simplicial~complexes.

It turns out that the problem remains undecidable even for geodesic agreement. To show this, we need to do some more work, but we can reduce this problem to the decision problem for clique agreement by a simple local transformation applied to the input graph.

\subsection{Implications for message-passing models}

Although we have focused above on asynchronous shared-memory models, our results also imply new impossibility results for asynchronous and synchronous message-passing systems.

\subparagraph{Asynchronous message-passing systems.}
In the asynchronous case, we obtain new impossibility results by using the asynchronous message-passing variant of the $t$-resilient asynchronous computability theorem (\Cref{thm:t-res-computability-msg}).
Together with \Cref{thm: SAthm} and \Cref{coro:ca}, we get new results on  resilience thresholds for asynchronous message-passing algorithms.

\begin{corollary}
  Consider the asynchronous message-passing model with $0 \le t < n/2$ crash faults and $n$ processes. Then the following are true:
  \begin{enumerate}
    \item Simplex agreement on $\A$ is solvable if and only if $\A$ is $(t-1)$-connected.
    \item Clique and monophonic agreement on $G$ are solvable if and only if $\K(G)$ is $(t-1)$-connected.
  \end{enumerate}
\end{corollary}

\subparagraph{Round lower bounds for synchronous message-passing.}
In synchronous systems with $f < n$ crash faults, approximate agreement tasks on graphs are always solvable -- simply by running exact convex agreement~\cite{nowak2019byzantine,constantinescu2024convex}.
For geodesic  agreement, the problem is
always solvable in $f/2 + O(\log D/\log \log D)$ rounds:
one can first run 2-set agreement and then approximate agreement on a path~\cite{alistarh2023wait}.
Recently, Fuchs et al.~\cite{fuchs2026optimal} showed that $f$-resilient clique agreement on \emph{any} diameter-$D$ graph requires $\Omega(\log D / \log \log D)$ rounds when $f \in \Theta(n)$.

On the other hand,
Alistarh et al.~\cite{alistarh2023wait} showed that for any graph that admits an ``impossibility labelling'' there is no $f$-resilient clique agreement protocol that solves the problem in $f/2$ rounds.
Our new characterisation implies round lower bounds for general graphs -- which do not necessarily admit such labellings -- in terms of the connectivity of their clique complex.
This can be done by lifting the lower bounds for $t$-resilient asynchronous systems to synchronous message-passing systems using \Cref{thm:round-by-round}.

\begin{corollary}
  Let $0 < t < f < n$ such that $f+t < n$.
  In the synchronous message-passing model  with $n$ processes and $f$ crash failures, we have the following bounds:
  \begin{enumerate}[(a)]

    \item If $\A$ is not $(t-1)$-connected, then simplex agreement on $\A$ requires at least $f/t$ rounds.

    \item If $\K(G)$ is not $(t-1)$-connected, then clique agreement on $G$ requires at least $f/t$~rounds.
  \end{enumerate}
\end{corollary}

\section{Characterisation of solvable approximate agreement tasks}\label{sec:main-tool}

In this section, we prove our main technical result about $\C$-convex simplex agreement. We split the proof of the next theorem into two parts, one for each direction.

\maintool*

\subsection{The first direction of \texorpdfstring{\Cref{thm:main-tool}}{Theorem 6}: (a) implies (b)}

Throughout this subsection, fix $(\Deltava, \A, \PhiC)$ and $f$ to be as in the statement of \Cref{thm:main-tool}.
Recall that a map $\iota \colon A \to B$ between topological spaces is
a \emph{$t$-equivalence} if the induced~map
\[
  \iota_\ast^{(m)}: \pi_m(A) \overset{\sim}{\rightarrow} \pi_m(B)
\]
is an isomorphism between the homotopy groups for all $0 \le m < t$, and a surjection for $m = t$. We need a
result from homotopy theory of simplicial complexes; see, e.g.,~\cite[Corollary 4.12]{Hatcher}.

\begin{lemma}
  \label{lemma:inclusion-is-equiv}
  Let $\A$ be a simplicial complex and $t \ge 0$. The inclusion map \(
    j\colon |\skelalphA| \rightarrow |\A|
  \)
  is a $t$-equivalence.
\end{lemma}

\subparagraph{Proof strategy.}
By \Cref{def:convex-agreement}, the carrier map of $\C$-convex simplex agreement is
\[
  \PhiC(\sigma) :=
  \begin{cases}
    \Delta^\sigma & \text{if }  \sigma \in \A \\
    \A[\langle \sigma \rangle_\C] & \text{otherwise.}
  \end{cases}
\]
The idea is to consider a fixed $\sigma \in \Deltava$ and the following diagram:
\[
  \begin{tikzcd}
    & {|\skelalphDeltaPhisig|} \\
    {|\skel^t \PhiC(\sigma)|} && {|\PhiC(\sigma)|.}
    \arrow["g", from=1-2, to=2-3]
    \arrow["{\iota}", from=2-1, to=1-2]
    \arrow["{j}"', from=2-1, to=2-3]
  \end{tikzcd}
\]
Here, $\iota$ and $j$ are the natural inclusion maps, and $g$ is the restriction of $f$ to $|\skelalphDeltaPhisig| \subseteq |\skelalphDeltava|$.
Since $\PhiC$ is defined in terms of the closure operator $\langle \cdot \rangle_\C$, one can verify using the definition of $\PhiC$ that  $\PhiC(V(\PhiC(\sigma))) = \PhiC(\sigma)$  for
all $\sigma \in \Deltava$.
Hence, the image of $|\skelalphDeltaPhisig|$ under $f$ indeed lies in
$|\PhiC(\sigma)|$ because $f$ is carried by $\PhiC$.
The argument boils down to first showing that $j \he f \circ \iota$ and then using the fact that $|\skel^t \DeltaPhisig|$
is a $(t-1)$-connected space. The proof is then concluded by the next~lemma.

\begin{lemma}\label{lemma:equivalences-and-connectedness}
  Let $t \geq 0$ and $A,B,C$ be non-empty topological spaces.  Suppose the diagram
  \[
    \begin{tikzcd}
      & {B} \\
      {A} && {C.}
      \arrow["g", from=1-2, to=2-3]
      \arrow["{\iota}", from=2-1, to=1-2]
      \arrow["{j}"', from=2-1, to=2-3]
    \end{tikzcd}
  \]
  commutes up to homotopy, i.e., $j$, $\iota$, and $g$ are continuous maps that satisfy $j \simeq g \circ \iota$. If $j$ is a
  $t$-equivalence and $B$ is $(t-1)$-connected, then $C$ is $(t-1)$-connected.
\end{lemma}
\begin{proof}
  Because $B$ is $(t-1)$-connected, we have by definition $\pi_m(B) = 0$ for each $0 \le m < t$.
  Hence, the induced map $\iota_\ast^{(m)} \colon \pi_m(A) \rightarrow \pi_m(B)$
  is
  the trivial homomorphism, i.e., $\iota_\ast^{(m)} = 0$ for each $0 \le m < t$.
  Since $j \he g \circ \iota$, it follows that the induced map $j_\ast^{(m)}$ can be computed as
  \[
    j_\ast^{(m)} = (g \circ \iota)_\ast^{(m)} = g_\ast^{(m)} \circ  \iota_\ast^{(m)} = g_\ast^{(m)} \circ 0 = 0,
  \]
  which is the trivial homomorphism.
  The claim now follows as $j$ is a $t$-equivalence, so by definition
  \(
    j_\ast^{(m)}: \pi_m(A) \rightarrow \pi_m(C)
  \)
  is an isomorphism for $0 \le m < t$.
  Therefore, we get that  $\pi_m(C)$ is trivial for each $0 \le m < t$, so $C$ is $(t-1)$-connected.
\end{proof}

\begin{lemma}\label{lemma:homotopy-eq}
  Let $\sigma\in \Deltava$ and
  \(
    g \colon |\skelalphDeltaPhisig| \to |\PhiC(\sigma)|
  \)
  be a continuous map carried by $\PhiC$.
  Suppose
  $\iota \colon |\skel^t\PhiC(\sigma) |\rightarrow |\skelalphDeltaPhisig|$ and
  $j \colon |\skel^t\PhiC(\sigma)| \rightarrow |\PhiC(\sigma)|$
  are the natural inclusion maps.
  Then $j \he g \circ \iota$.
\end{lemma}
\begin{proof}
  To prove the claim, we show that the straight-line homotopy
  \begin{align*}
    H: [0,1] \times |\skelalphPhisig| &\rightarrow |\PhiC(\sigma)|\\
    (s,x) &\mapsto (1-s)x + sg(x)
  \end{align*}
  is a well-defined continuous homotopy between $j$ and $g \circ \iota$. It is clear that if $H$ is well-defined, then it is also continuous because $g$ is continuous. Since
  \[
    \bigcup_{\tau \in \skel^t\PhiC(\sigma)} |\tau| = |\skel^t\PhiC(\sigma)|,
  \]
  it suffices to prove that $H$ is well-defined on each simplex $|\tau| \subseteq |\skel^t \PhiC(\sigma)|$.
  By assumption $g$ is carried by $\PhiC$.
  By definition of convex simplex agreement, the carrier map $\PhiC$ satisfies $\PhiC(\tau) = \Delta^\tau$ for each $\tau \in \A$.
  Therefore,
  for any $\tau \in \skel^t \PhiC(\sigma)$, we have
  \begin{equation*}
    g(|\tau|)  \subseteq |\PhiC(\tau)| = |\Delta^\tau| = |\tau|.
  \end{equation*}
  That is, the image of $|\tau|$ under $g$ stays in $|\tau|$. Recall that the geometric realisation  $|\tau| \subseteq |\PhiC(\sigma)|$ of any simplex $\tau$ is a convex subset of $\R^{d+1}$.
  Hence, for any point $x \in |\tau|$, the whole line segment $\Set{H(s,x)}{s \in [0,1]}$ lies in $|\tau|$. Therefore, for all $s \in [0,1]$ we have
  \(
    H(s, |\tau|) \subseteq |\tau|
  \),
  so $H$ is indeed a well-defined homotopy on each simplex $|\tau| \subseteq |\skelalphPhisig|$.
\end{proof}

\begin{lemma}
  Let $(\Deltava, \A, \PhiC)$ be the $\C$-convex simplex agreement task on $\A$ and
  \[
    f \colon |\skel^t \Deltava| \rightarrow |\A|
  \]
  a continuous map carried by $\PhiC$. Then $\PhiC(\sigma)$ is $(t-1)$-connected for each  $\sigma \in  \Deltava$.
\end{lemma}
\begin{proof}
  Let $\sigma \in \Deltava$ be a simplex.
  Recall that the image of $|\skelalphDeltaPhisig|$ under $f$ lies in
  $|\PhiC(\sigma)|$ because $f$ is carried by $\PhiC$.
  Hence, we can consider the restriction $g \colon |\skelalphDeltaPhisig| \to |\PhiC(\sigma)|$ of the map $f$.
  Consider the following diagram
  \[
    \begin{tikzcd}
      & {|\skelalphDeltaPhisig|} \\
      {|\skel^t \PhiC(\sigma)|} && {|\PhiC(\sigma)|.}
      \arrow["g", from=1-2, to=2-3]
      \arrow["{\iota}", from=2-1, to=1-2]
      \arrow["{j}"', from=2-1, to=2-3]
    \end{tikzcd}
  \]
  Because the standard simplex $\DeltaPhisig$ is contractible, it is also $(t-1)$-connected. By
  \Cref{lemma:inclusion-is-equiv}, the inclusion map $h \colon |\skelalphDeltaPhisig| \rightarrow|\DeltaPhisig|$
  is a $t$-equivalence. Hence, $|\skelalphDeltaPhisig|$ is $(t-1)$-connected.

  By
  \Cref{lemma:inclusion-is-equiv}, the inclusion map $j \colon |\skelalphPhisig| \rightarrow |\PhiC(\sigma)|$ is a $t$-equivalence.
  Consider the inclusion map
  $\iota \colon |\skel^t \PhiC(\sigma)| \to |\skelalphDeltaPhisig|$. \Cref{lemma:homotopy-eq} shows that the map $j$ is homotopic to $g \circ \iota$. Applying  \Cref{lemma:equivalences-and-connectedness} to $\iota, g, j$ shows that $\PhiC(\sigma)$ is $(t-1)$-connected.
\end{proof}

\subsection{The second direction of \texorpdfstring{\Cref{thm:main-tool}}{Theorem 6}: (b) implies (a)}

For the second direction, we need another  elementary result from algebraic topology,
which follows from the fact that continuous maps from simplicial complexes to $(t-1)$-connected complexes can be extended to a larger domain; see, e.g., \cite[Lemma 4.7]{Hatcher}.

\begin{lemma}
  \label{lemma:extension}
  Let $\A$ be a $(t-1)$-connected simplicial complex, $m \leq t$, and
  $f \colon |\partial \Delta^m| \rightarrow |\A|$ be a continuous map.
  Then there exists a continuous
  map $\Tilde{f} \colon |\Delta^m| \rightarrow |\A|$ such that $\Tilde{f}_{\mid |\partial \Delta^m|} = f$.
\end{lemma}

We are now ready to prove the other direction of \Cref{thm:main-tool}.

\begin{lemma}
  Let $(\Deltava, \A, \PhiC)$ be the $\C$-convex simplex agreement task on $\A$ such that $\PhiC(\sigma)$ is
  $(t-1)$-connected for each $\sigma \in  \Deltava$. Then there exists a continuous map
  \(
    f \colon |\skel^t \Deltava| \rightarrow |\A|
  \)
  carried by $\PhiC$.
\end{lemma}
\begin{proof}

  Define $\B_m := \skel^m \Deltava \cup \skel^t\A$ for each $0 \le m \le t$.
  We prove the claim by inductively constructing for each $0 \le m \le t$ a continuous map
  \(
    f_m \colon |\B_m| \rightarrow |\A|
  \)
  carried by $\PhiC$.
  The claim then follows by choosing $f := f_t$, as $\B_t = \skelalphDeltava$.
  For the base case $m=0$, we can choose $f_0$ to be the inclusion map $f_0 \colon | \skel^t \A|   \rightarrow |\A|$, because
  \[
    \B_0 = \skel^0\Deltava \cup \skel^t\A = V(\A) \cup \skel^t\A = \skel^t\A.
  \]
  Since
  $f_0$ is an inclusion map of topological spaces, it is continuous. In addition, for any $\sigma \in \skel^t \A$, we get that
  \[
    f_0(|\sigma|) = |\sigma| = |\Delta^\sigma| = |\PhiC(\sigma)|,
  \]
  where the last equality follows from the property that $\PhiC(\sigma) = \Delta^\sigma$ for each $\sigma \in \A$.
  Thus, $f_0$ is carried by $\PhiC$.

  For the inductive step, suppose that for some $0 \le m < t$ there exists a continuous map
  \(
    f_m \colon |\B_m| \rightarrow |\A|
  \)
  carried by $\PhiC$. Consider the simplices $\sigma_1, \ldots, \sigma_s \in  \B_{m+1} \setminus \B_{m}$.
  Because
  \(
    \partial \sigma_i = \Delta^{\sigma_i} \setminus \{\sigma_i\} \subseteq \B_m,
  \)
  the map  $f_m$ is already defined on $|\partial \sigma_i|$ for each $1 \le i \le s$. Since $f_m$ is carried by $\PhiC$, we have by monotonicity of $\PhiC$ for each $1 \le i \le s$ that
  \[
    f_m(|\partial \sigma_{i}|) \subseteq \bigcup_{\tau \in \partial \sigma_i} f_m(|\tau|) \subseteq \bigcup_{\tau \in \partial \sigma_i}|\PhiC(\tau)| \subseteq |\PhiC(\sigma_i)|.
  \]
  In particular, the restriction of $f_m$ to $|\partial \sigma_i|$ is a continuous map $g_i \colon |\partial \sigma_i| \to |\PhiC(\sigma_i)|$.
  By assumption $\PhiC(\sigma_{i})$ is $(t-1)$-connected, so by \Cref{lemma:extension} there exists an extension
  $h_i \colon |\sigma_{i}| \to |\PhiC(\sigma_i)|$ of $g_i$ for each $1 \le i \le s$.
  Because
  \[
    |\B_{m+1}| = |\B_m| \cup \bigcup_{i=1}^s |\sigma_i|,
  \]
  we can define the desired map $f_{m+1} \colon |\B_{m+1}| \to |\A|$ as
  \[
    f_{m+1} =
    \begin{cases}
      h_i(x) & \text{if } x \in |\sigma_i| \setminus |\partial \sigma_i| \text{ for some } 1 \le i \le s \\
      f_m(x) & \text{otherwise.}
    \end{cases}
  \]
  Since $f_{m+1}$ is defined on a partition of $|\B_{m+1}|$, it is well-defined. The map $f_{m+1}$ is continuous because $|\B_m| \cup \bigcup_{i=1}^s |\sigma_i|$
  is a finite covering by closed sets
  of $|\B_{m+1}|$ and $f_{m+1}$ is continuous on each component of the covering.

  Finally, it remains to verify that $f_{m+1}$ is carried by $\PhiC$. Let $\sigma \in \B_{m+1}$. If $\sigma \in \B_{m}$, then by the induction hypothesis we get that
  \(
    f_{m+1}(|\sigma|) = f_{m}(|\sigma|) \subseteq |\PhiC(\sigma)|.
  \)
  Otherwise, if $\sigma \in \B_{m+1} \setminus \B_m$, then $\sigma = \sigma_i$ for some $1 \le i \le s$. Since $h_i$ extends $g_i$ to $|\sigma_i|$, we have
  \(
    f_{m+1}(|\sigma_i|) = h_i(|\sigma_i|) \subseteq |\PhiC(\sigma)|
  \). Thus, $f_{m+1}$ is carried by $\PhiC$.
\end{proof}

\section{Applications}\label{sec:applications}

In this section, we apply \Cref{thm:main-tool} to simplex agreement and the different approximate agreement tasks on graphs.
We give the characterisation of solvability of simplex agreement and clique agreement.
We prove that clique agreement is $t$-resilient solvable on a graph $G$ if and only if monophonic agreement is $t$-resilient solvable on $G$. Additionally, we show that geodesic
agreement is in general strictly harder to solve than the clique and monophonic agreement tasks,
but always wait-free solvable on dismantlable graphs.

\subsection{Simplex and clique agreement}

\sathm*
\begin{proof}
  Simplex agreement on $\A$ is $\C$-convex simplex agreement, where  the convexity space $\C$ is given by $\C = \A \cup \{\emptyset, V(\A)\}$.
  Since $\A \subseteq \C$, the task is $t$-resilient solvable if and only if $\A[S]$ is $(t-1)$-connected for any non-empty $S \in \C$ by \Cref{coro:convex-agreement-solvability}.
  The only non-empty convex sets in $\C$ are the simplices of $\A$ and $V(\A)$.
  Each simplex of $\A$ is trivially contractible.
  Hence, by \Cref{coro:convex-agreement-solvability}, the task is $t$-resilient solvable if and only if $\A$ is $(t-1)$-connected.

  By definition, a protocol is wait-free if it is $(n-1)$-resilient. Thus,
  there exists a wait-free protocol for all $n \ge 1$ if and only if $\A$ is $(n-2)$-connected for all $n \ge 1$.
  By \Cref{Whitehead} the simplicial complex $\A$ is contractible if and only if $\A$ is $(t-1)$-connected for all $t \ge 0$.
\end{proof}

As observed by Ledent~\cite{ledent2021brief}, clique agreement on a graph $G$ is the same task as simplex agreement on its clique complex $\K(G)$. Therefore, from \Cref{thm: SAthm} we immediately get the analogous result for clique agreement and a solution to Ledent's conjecture.

\cacorollary*

\subsection{Solvability of monophonic agreement}

We now show that monophonic agreement is solvable on the same class of graphs as clique agreement.
Recall that monophonic agreement on a graph $G$ is the colourless task, where the
carrier map takes any set of vertices $\sigma \subseteq V(G)$ to the subcomplex induced by the monophonic convex hull of $\sigma$ in the graph $G$.
Therefore, this is precisely the $\C$-convex task on the clique complex $\K(G)$, where $\C$ is the monophonic convexity of $G$.

Since cliques are monophonically convex, we have $\K(G) \subseteq \C$. Therefore,
\Cref{coro:convex-agreement-solvability} yields that
monophonic agreement
is $t$-resilient solvable if and only if the subcomplex $\K(G[S])$ of $\K(G)$ induced by $S$ is $(t-1)$-connected for any monophonically convex set $S \in \C$.

Somewhat surprisingly, this condition can be further simplified. Instead of requiring that \emph{all} monophonically convex subsets are $(t-1)$-connected, it is equivalent to check only if the whole
complex $\K(G)$ is $(t-1)$-connected. This is captured by the following result:

\monothm*

The proof relies on a standard topological fact which can be found, e.g., in the book by Kozlov~\cite[Corollary 6.30]{CombinatorialTopoKozlov}.

\begin{lemma}
  \label{lem:gluing}
  Let $\A$ be a simplicial complex, and let $\B_1, \B_2 \subseteq \A$ be subcomplexes such that
  \begin{equation*}
    \A = \B_1 \cup \B_2 \quad \text{ and } \quad \B = \B_1 \cap \B_2.
  \end{equation*}
  If $\A$ and $\B$ are $t$-connected for $t \ge 0$, then $\B_1$ and $\B_2$ are $t$-connected.
\end{lemma}

\begin{proof}[Proof of \Cref{thm:monophonic-sets-are-connected}]
  Let $S \subseteq V(G)$ be a non-empty monophonically convex set in $G$.
  Let $C_1, \ldots, C_k$
  be the vertex sets of the connected components of the graph $G[V \setminus S]$. Since $\K(G)$ is at least 0-connected, its 1-skeleton $\skel^1\K(G) = G$ is connected. Hence, the vertex sets
  \begin{equation*}
    B_i := N_G(C_i) \cap S
  \end{equation*}
  are non-empty. We claim that each $B_i$ is a clique in the graph $G$. Indeed, suppose that $a,b \in B_i$ are two non-adjacent vertices. Since $C_i$ is connected, there exists a path from $a$ to $b$
  whose internal vertices all lie in $C_i$. Among all such paths choose a path $P$ of minimum length. Then $P$ is induced. Indeed, if there were an edge between two non-consecutive vertices of $P$, one would
  obtain a shorter path from $a$ to $b$ with at least one internal vertex in $C_i$, since by assumption $a$ and $b$ are non-adjacent. This contradicts the monophonic convexity of $S$. Therefore, $a$ and $b$
  must be adjacent and $B_i$ is a clique.

  Now consider the induced subgraphs
  \begin{equation*}
    G_0 := G[S] \qquad G_i := G[S \cup C_1 \cup \cdots \cup C_i] \qquad  H_i := G[B_i \cup C_i]
  \end{equation*}
  for $1 \leq i \leq k$.
  Then in particular we have $G_k = G$.
  It is easy to see that
  \begin{equation*}
    \K(G_{i-1}) \cup \K(H_i) = \K(G_i) \quad \text{and} \quad \K(G_{i-1}) \cap \K(H_i) = \K(G[B_i]).
  \end{equation*}
  Since $B_i$ is a non-empty clique, the clique complex $\K(G[B_i])$ is the standard simplex, which is contractible, and hence $t$-connected.
  Now apply \Cref{lem:gluing} to
  \begin{equation*}
    \A := \K(G_i), \quad \B_1 := \K(G_{i-1}), \quad \B_2 := \K(H_i).
  \end{equation*}
  It follows for all $1 \leq i\leq k$ that if $\K(G_i)$ is $t$-connected then $\K(G_{i-1})$ is also $t$-connected.
  Since $\K(G_k) = \K(G)$ is $t$-connected, we get that $\K(G_0) = \K(G[S])$ is $t$-connected.
\end{proof}

\cliqueminimalequiv*
\begin{proof}
  Clearly, any monophonic agreement protocol solves clique agreement, so we only need to consider the other direction.
  Suppose clique agreement is  $t$-resilient solvable on $G$.
  Then by \Cref{coro:ca} the clique complex $\K(G)$ is $(t-1)$-connected.
  Then \Cref{thm:monophonic-sets-are-connected} implies that the subcomplex $\K(G[S])$ is $(t-1)$-connected for each monophonically convex set $S$ of $G$.
  \Cref{coro:convex-agreement-solvability} implies that monophonic agreement is $t$-resilient solvable on $G$.
\end{proof}

The Whitehead theorem (see \Cref{Whitehead}) implies that monophonically convex subsets of a contractible clique complex
are contractible. However, we note that \Cref{thm:monophonic-sets-are-connected} is no longer true if we replace $\K(G)$ by an arbitrary simplicial complex $\A$. For example,
the non-flag complex $\Delta^{\{a,b,c,d\}} \setminus \{ \{a,b,c,d\}, \{a,b,c\}\}$ gives a counter-example.

\subsection{Separation between monophonic and geodesic agreement}

Observe that cliques are also geodesically convex.
Thus, by the same reasoning as for monophonic agreement in the previous section,  \Cref{coro:convex-agreement-solvability}
shows that geodesic agreement on $G$ is $t$-resilient solvable if and only if every geodesically convex subset of $G$ induces a $(t-1)$-connected subcomplex of $\K(G)$.

As an immediate corollary of this observation,
we get the following separation of geodesic agreement from the clique and monophonic agreement tasks.

\begin{corollary}
  There exists a graph $G$ on which clique agreement admits a wait-free 3-process protocol, whereas geodesic agreement does not.
\end{corollary}

\begin{proof}
  Assume we have 3 processes that start with inputs on the vertices $a,b,c$ on the graph $G$ depicted below:

  \GraphAHighlighted

  \noindent
  The geodesic convex hull of the inputs $\set{a,b,c}$ is the set $S = \set{a,b,c,d}$ which induces the 4-cycle $G[S]$.
  It is well-known that geodesic agreement cannot be reached on cycles of length at least 4 (see \cite{alistarh2023wait}).
  Put otherwise,
  the clique complex $\K(G[S])$ of the geodesically convex set $S$ is now homeomorphic to the non-contractible 1-dimensional sphere.
  At the same time it is easy to see that $\K(G)$ is contractible since it is a
  triangulation of the square $[0,1] \times[0,1] \subset \R^2$ and hence, by \Cref{coro:ca}, clique agreement is wait-free solvable for $n=3$ processes.
\end{proof}

\subsection{Geodesic agreement is wait-free solvable on dismantlable graphs}

In this section, we show that the class of dismantlable graphs admits a wait-free solution for the geodesic agreement task. For this, we need the following lemma that shows geodesically convex sets of dismantlable graphs induce dismantlable subgraphs. Note that the same is not true for arbitrary induced subgraphs of a dismantlable graph.

\begin{restatable}{lemma}{lemmadismantlable}\label{lemma:dismantlable-convex-hulls}
  Let $G = (V,E)$ be a dismantlable graph and $S \subseteq V$.
  If $S$ is  geodesically convex in $G$, then
  the induced subgraph $G[S]$ is dismantlable.
\end{restatable}

The above lemma seems to be a folklore result that follows from the fact that dismantlable graphs are closed under retractions~\cite{nowakowski1983vertex}. We give a simple self-contained proof in Appendix~\ref{apx:dismantlable}.

The next step is to observe that the clique complex $\K(G)$ of any dismantlable graph $G$ is contractible~\cite{larrion2008contractibility}, so geodesically convex sets of $G$ induce contractible subcomplexes of $\K(G)$. Therefore, \Cref{coro:convex-agreement-solvability} and \Cref{lemma:dismantlable-convex-hulls} together imply that geodesic agreement is wait-free solvable on dismantlable graphs.

\corodis*

\section{Undecidability results}

In general,  it is undecidable to determine if a $t$-resilient protocol for a given task exists in the asynchronous shared-memory model with read-write registers~\cite{gafni1998three,herlihy1997decidability}. For example, loop agreement tasks~\cite{herlihy2003classification} are colourless tasks for which the solvability question is undecidable.
Alcántara et al.~\cite{alcantara2019topology} showed that  it is undecidable to determine if clique agreement (1-gathering in their terminology) has a wait-free protocol for $n=3$ processes.

However, it is not immediately obvious if the solvability question remains undecidable when considering (a) non-wait-free protocols, and (b) the more constrained monophonic and geodesic agreement tasks.
In this section, we show that this remains the case.

\begin{theorem}\label{thm:undecidabilities}
  Let $t \geq 2$ and $n \geq t+1$.
  It is undecidable to determine if there exists an $n$-process $t$-resilient asynchronous shared-memory protocol for any of the following tasks:
  \begin{enumerate}[(a)]
    \item simplex agreement on a given simplicial complex $\A$,
    \item clique agreement on a given graph $G$,
    \item monophonic agreement on a given graph $G$, and
    \item geodesic agreement on a given graph $G$.
  \end{enumerate}
\end{theorem}

We start by proving the first three results in \Cref{sec:clique-undecidability}. The last result requires some more work, so we deal with it separately in \Cref{sec:geodesic-undecidability}.
Throughout, we only consider connected graphs, as clearly, approximate agreement tasks are not solvable otherwise.

\subsection{Undecidability results for simplex, clique and monophonic agreement}\label{sec:clique-undecidability}

It is well-known that determining whether a given simplicial complex is $m$-connected for any $m \geq 1$ is  undecidable. This is because it reduces to the group
triviality problem, which is undecidable; see \cite{groupunsolv}. For an explicit construction that can be used for the reduction, we refer to \cite[Proposition 1.26]{Hatcher}.

\begin{theorem}\label{thm:undecidable-connected}
  Let $m \ge 1$.
  It is undecidable to determine if a simplicial complex is $m$-connected.
\end{theorem}

This immediately implies the first result.

\begin{lemma}
  Solvability of simplex agreement is undecidable for any $t \ge 2$ and $n > t$.
\end{lemma}
\begin{proof}
  By \Cref{thm: SAthm}, simplex agreement has a $t$-resilient $n$-process protocol if and only if $\A$
  is $(t-1)$-connected.
  By \Cref{thm:undecidable-connected}, we know that $(t-1)$-connectedness is undecidable for all $t \geq 2$, and hence, so is the solvability of the $t$-resilient simplex agreement.
\end{proof}

Recall that the \emph{barycentric subdivision} $\Bary\A$ of a simplicial complex $\A$ is the abstract simplicial complex defined as follows: The vertices of $\Bary \A$ are the simplices of $\A$ and a subset $\set{\sigma_0, \ldots, \sigma_h} \subseteq \A$ is a simplex of $\Bary\A$ if and only if the simplices can be ordered as a chain of subsets
\(
  \sigma_0 \subsetneq \sigma_1 \subsetneq \ldots \subsetneq \sigma_h
\).

\begin{lemma}\label{lemma:ca-undecidablility}
  Solvability of clique agreement is undecidable for any $t \ge 2$ and $n > t$.
\end{lemma}
\begin{proof}
  We reduce deciding $(t-1)$-connectivity of $\A$ to deciding if clique agreement is $t$-resilient solvable on a given graph.
  Let $\A$ be an arbitrary simplicial complex.
  Consider the underlying graph of its barycentric subdivision, i.e., the graph $\skel^1\Bary\A$.
  This is the graph $G=(V,E)$ defined by $V = V(\Bary \A)$ and $E = \{ \sigma \in \Bary \A : \dim \sigma = 1\}$.
  We claim~that
  \begin{equation*}
    \K(G) = \Bary \A .
  \end{equation*}
  Indeed, it suffices to prove that for every subset of vertices
  \[
    X \subseteq V(\Bary\A)
  \]
  the set $X$ is a clique in the graph $G$ if and only if $X$ is a simplex in $\Bary\A$. The ``if''-direction follows immediately by definition, since the
  edges of every simplex form a clique.

  For the other direction, assume $X$ is a clique in $G$.
  We show that $X$ is a simplex in $\Bary\A$.
  By definition of barycentric subdivision,
  a set $X = \set{\sigma_0, \ldots, \sigma_h} \subseteq \A$ is a simplex in $\Bary\A$ if one can
  relabel the indices to get a chain of inclusions
  \(
    \sigma_0 \subsetneq \sigma_1 \subsetneq\ldots \subsetneq \sigma_h
  \).
  Since every pair $\sigma, \sigma' \in X \subseteq \A$ forms an edge in $\Bary\A$, we have either $\sigma \subseteq \sigma'$ or $\sigma' \subseteq \sigma$.
  Hence, the set $\Set{\sigma_i}{0 \le i \leq h}$ is totally ordered by the inclusion relation.
  Thus, $X \in \Bary \A$.

  To complete the reduction, we have established that $\K(G) = \Bary \A$.
  We know that clique agreement is $t$-resilient solvable if and only if $\K(G)$ is $(t-1)$-connected by \Cref{coro:ca}.
  It is well-known that $\Bary \A \he \A$, that is, barycentric subdivision preserves the homotopy type. Therefore, any procedure for deciding if clique agreement is $t$-resilient solvable
  would decide if $\A$ is $(t-1)$-connected, which is not possible by \Cref{thm:undecidable-connected}.
\end{proof}

\begin{lemma}
  Solvability of monophonic agreement is undecidable for $t \ge 2$ and $n > t$.
\end{lemma}
\begin{proof}
  By \Cref{coro:clique-minimal-paths-equivalence} monophonic agreement has a  $t$-resilient $n$-process protocol on $G$ if and only if clique agreement has
  one on $G$. Hence, monophonic agreement is also undecidable.
\end{proof}

\subsection{Undecidability result for geodesic agreement}\label{sec:geodesic-undecidability}

We now reduce the problem of deciding the existence of clique agreement protocols to the problem of deciding the existence of geodesic agreement protocols.
Without loss of generality, it suffices to consider only connected graphs, as geodesic agreement can never be solved on disconnected graphs.
For the reduction, we transform any connected graph $G$ to a graph $H$ with the following properties:
\begin{enumerate}
  \item clique agreement and geodesic agreement are the same task on $H$, and
  \item $\K(G)$ and $\K(H)$ are homotopy equivalent.
\end{enumerate}

\newcommand{\vzero}{v^{0}}
\newcommand{\vone}{v^{1}}
\newcommand{\vtwo}{v^{2}}

\subparagraph{The construction.}
Let $G$ be any connected input graph.
We construct the graph $H$ from $G$ by replacing edges by a simple gadget
as follows:
\begin{itemize}
  \item Each vertex $v \in V(G)$ is replaced with a path $P_v = \{\vzero, \vone, \vtwo \}$ of length 2 from $\vzero$ to $\vtwo$.
  \item For each edge $\{u,v\} \in E(G)$, we connect in $H$ each vertex of $P_u$ with all vertices of $P_v$.
\end{itemize}
\begin{figure}
  \centering
  \includegraphics[width=\textwidth]{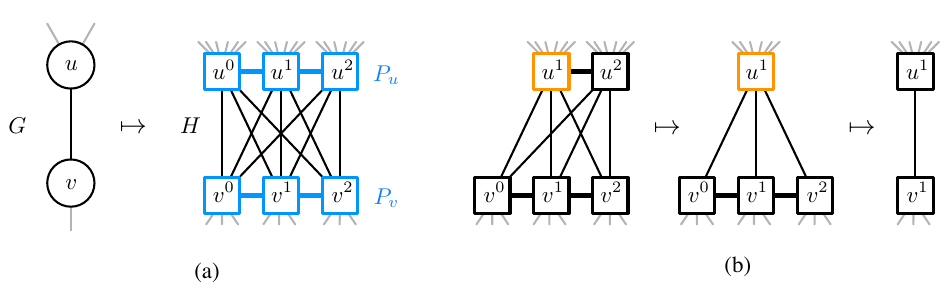}
  \caption{Gadget for the undecidability reduction. (a)
    The graph $H$ is constructed from $G$ by replacing each vertex $v$ of $G$ by a path $P_v$. For any edge $\{u,v\}$ in $G$, we add all edges from $P_u$ to $P_v$. (b) Local collapses in $H$. The vertex $u^1$ dominates
    $u^0$ in $H$, so we can remove the vertex $u^0$ from $H$ without changing the homotopy type of its clique complex. The same applies to $u^2$ after removing $u^0$. Similarly, we can remove $\vzero$ and $\vone$ for each $v$ to obtain a graph isomorphic to $G$.
  }\label{fig:undecidability}
\end{figure}
\Cref{fig:undecidability}a illustrates the construction.
Since $G$ is connected, the graph $H$ is also connected.
Next, we prove that the geodesically convex sets in $H$ are precisely cliques or $V(H)$.
This implies that the clique agreement and geodesic agreement tasks are the same task on $H$.

\begin{lemma}\label{lemma:geodesic-sets-in-h}
  Let $S$ be a geodesically convex set in $H$.
  If $S$ is not a clique, then $S = V(H)$.
\end{lemma}
\begin{proof}
  Observe that if $P_v \subseteq S$ for some $v \in V(G)$, then $P_u \subseteq S$ for any $u \in N_G(v)$. This is because for each $z \in P_u$, the path $\vzero, z, \vtwo$ is a shortest path between $\vzero$ and $\vtwo$ in $H$.
  In particular, it follows by induction that if $P_v \subseteq S$, then $S = V(H)$.
  Thus, to prove the lemma, it suffices to show that $P_v \subseteq S$ for some $v \in V(G)$ whenever $S$ is not a~clique.

  Suppose that $S$ is not a clique. Then there exist non-adjacent vertices $x,y \in S$.
  Consider the following two cases:
  \begin{enumerate}
    \item Suppose that $x,y \in P_v$ for some $v \in V(G)$. Because $x$ and $y$ are non-adjacent in $H$, we have that $\{x,y\} = \{\vzero, \vtwo\}$. Because $S$ is geodesically convex, it
      is closed under shortest paths, so we get that $\vone \in S$. This implies that $P_v \subseteq S$.
    \item Suppose that $x \in P_u$ and $y \in P_v$ for some $u \neq v$.
      Because $x$ and $y$ are non-adjacent, so are $u$ and $v$.
      Let $x=z_0, \ldots, z_\ell = y$ be a shortest path of length $\ell \ge 2$ in $H$.
      Consider the walk $u=w_0, \ldots, w_\ell = v$ in $G$, where $z_i \in P_{w_i}$.
      Observe that $w_i \neq w_{i+1}$ for each $0 \le i < \ell$.
      Otherwise, we would have $z_{i}, z_{i+1} \in P_{w_i}$, which would violate the minimality of the shortest path in $H$, as by construction of $H$ we could find a strictly shorter path.
      Now $z_0, z_1, z_2$ form a path of length 2.
      So each $q \in P_{w_1}$ forms a path $z_0, q, z_2$ of length 2 in $H$.
      Because $S$ is closed under shortest paths, we have $P_{w_1} \subseteq S$. \qedhere
  \end{enumerate}
\end{proof}

We say that a vertex $v$ is dominated in $\A$ if there exists a vertex $u$ such that all maximal simplices containing $v$ also contain $u$.
Let $\A-v$ denote the subcomplex of $\A$ obtained by deleting all simplices that contain $v$.
We make use of the fact that removing a dominated vertex does not change the homotopy type,
because this operation is a strong collapse; see, e.g.,~\cite[Section 2]{barmak2012strong}.

\begin{lemma}\label{lemma:domination-he}
  Let $\A$ be a simplicial complex.
  If a vertex $v$ is dominated in $\A$, then $\A \he \A - v$.
\end{lemma}

With the above lemma, we show that the clique complexes of $G$ and $H$ are homotopy equivalent.
The idea of the proof is illustrated in \Cref{fig:undecidability}b: for each $u \in V(G)$, we can remove vertices $u^0$ and $u^2$ of $H$ one by one without changing the homotopy type of the complex of cliques. Eventually, we end up with a graph that is isomorphic to $G$.

\begin{lemma}\label{lemma:g-and-h-are-he}
  The graphs $G$ and $H$ satisfy $\K(G) \he \K(H)$.
\end{lemma}
\begin{proof}
  Let $v_1, \ldots, v_k$ be the vertices of $G$.
  Define $H_0 = H$. For $0 \le i < k$, let $H_{i+1}$ be the graph obtained from $H_i$ by deleting the vertices $\vzero_{i+1}$ and $\vtwo_{i+1}$.
  Note that  $\vzero_{i+1}$ and $\vtwo_{i+1}$ are dominated by $\vone_{i+1}$ in $H_i$, and thus, also in $\K(H_i)$.

  Applying \Cref{lemma:domination-he} to the complex $\K(H_i)$ and $\vzero_{i+1}$, and then applying \Cref{lemma:domination-he} again  with $\vtwo_{i+1}$ to the resulting complex yields
  \[
    \K(H_{i}) \he \K(H_{i+1})
  \]
  for all $0 \le i < k$. Thus, $\K(H) \he \K(H_k)$ by iteration.
  Note that the map $\vone_i \mapsto v_i$ gives an isomorphism between $H_k$ and $G$.
  Because $H_k$ and $G$ are isomorphic, so are their clique complexes, which implies that $\K(H_k)\he \K(G)$. Hence, we get
  \[
    \K(H) \he \K(H_k)   \he \K(G). \qedhere
  \]
\end{proof}

\begin{lemma}
  The solvability of geodesic agreement is undecidable for any $t \ge 2$ and $n > t$.
\end{lemma}
\begin{proof}
  We reduce the decision problem for clique agreement (CA) to the decision problem for geodesic agreement (GA).
  Suppose there exists a procedure for deciding the existence of a $t$-resilient  geodesic agreement protocol for $n$ processes on connected graphs.
  Clearly, for any input graph $G$, we can compute $H$.
  Then by the above lemmas and \Cref{coro:ca}, we get that
  \begin{align*}
    \text{GA is solvable on $H$} &\iff \text{CA is solvable on $H$} && \text{(by \Cref{lemma:geodesic-sets-in-h})} \\
    &\iff \K(H) \text { is $(t-1)$-connected} && \text{(by \Cref{coro:ca})} \\
    &\iff \K(G) \text { is $(t-1)$-connected} && \text{(by \Cref{lemma:g-and-h-are-he})} \\
    &\iff \text{CA is solvable on $G$} && \text{(by \Cref{coro:ca})}.
  \end{align*}
  Thus, a procedure for deciding the existence of a $t$-resilient geodesic agreement protocol for $n$ processes on connected graphs would contradict \Cref{lemma:ca-undecidablility}.
\end{proof}

\section{Conclusions}

In this work, we gave a complete topological characterisation of the solvability of simplex agreement and commonly studied variants of approximate agreement on graphs. As a special case, we gave a proof of Ledent's conjecture~\cite{ledent2021brief} that had remained open for several years.
While we focused on clique, monophonic, and geodesic agreement on graphs, our main technical result applies  to other convexity spaces on graphs in which all cliques are convex.

Our characterisation also implies new impossibility results for message-passing models parameterised by the connectivity of the clique complex. In the asynchronous model, we obtained new resilience bounds, whereas in the synchronous model, we obtained new round lower bounds.
On the other hand, our positive results are non-constructive: on the solvable graph classes, we do not obtain concrete agreement protocols.
Indeed, the undecidability results imply that, e.g., the worst-case number of iterations in the layered snapshot model needed to solve clique, monophonic, or geodesic agreement is uncomputable.
It remains an interesting open problem to obtain protocols with optimal resilience and round complexity for graph classes beyond block graphs~\cite{fuchs2026optimal}.

\subparagraph{AI Disclosure:} We used GPT 5.0, GPT 5.4 and Opus 4.7 to assist with literature search,
proofreading the manuscript, preparing the illustrations in \Cref{fig:intro-graphs} and \Cref{fig:graphs}, and identifying proof strategies.
The tools suggested proof approaches, particularly for \Cref{thm:monophonic-sets-are-connected} and \Cref{lemma:dismantlable-convex-hulls} and the gadget used in \Cref{sec:geodesic-undecidability}.
However, these proofs were corrected and simplified by the authors.
The authors verified the correctness and originality of all content including~references.

\bibliography{references}

\newpage

\appendix

\section{Geodesic convex hulls of dismantlable graphs are dismantlable}
\label{apx:dismantlable}

In this appendix, we give a proof of the following lemma.

\lemmadismantlable*

Throughout, let $v_1, \ldots, v_k$ be the vertices of $G$ listed according to the dismantling order.
For each $1 \le i \le k$, let $G_i$ be the subgraph induced by $\{v_i, \ldots, v_k\}$.

\begin{lemma}\label{lemma:isometric}
  Let $u,v \in V(G_{i+1})$ for some $1 \le i < k$.
  If $\rho$ is a shortest path between $u$ and $v$ in $G_{i+1}$,
  then $\rho$ is a shortest path in $G_i$.
\end{lemma}
\begin{proof}
  For the sake of contradiction,
  suppose that $\rho$ is not a shortest path between $u$ and $v$ in $G_i$.
  Then there is a strictly shorter path $\rho'$ in $G_i$ from $u$ to $v$ that visits $v_i$.
  However, then we can find a path $\rho''$ of length at most the same as $\rho'$ contained entirely in $G_{i+1}$ that visits the dominator $w$ of $v_i$ instead of $v_i$, contradicting that $\rho$ was a shortest path in $G_{i+1}$.
\end{proof}

\begin{proof}[Proof of \Cref{lemma:dismantlable-convex-hulls}]
  For each $1 \le i \le k$, define
  \(
    S_i = S \cap \{v_i, \ldots, v_k\}.
  \)
  We show that if $S_i$ is convex in $G_i$, then $G_i[S_i]$ is dismantlable. The lemma follows by observing that $G_1[S_1] = G[S]$.
  The proof is by reverse induction on $i$.
  The base case $i=k$ is trivial, because $S_k \subseteq \{v_k\}$.
  For the inductive step, suppose the claim holds for some $2 \le i+1 \le k$.
  Let $S_i$ be a convex set in $G_i$.
  If $|S_i|=1$, then $G_i[S_i]$ is trivially dismantlable. Hence, assume $|S_i|>1$.

  If $v_i \notin S_i$, then $S_i =  S_{i+1}$ is also convex in $G_{i+1}$, because by \Cref{lemma:isometric} any shortest path contained in $S_{i+1}$ is also a shortest path in $G_i$. Hence, $S_{i+1}$ is closed under taking shortest paths also in $G_i$.
  Thus, $G_{i}[S_i] = G_{i+1}[S_{i+1}]$
  is dismantlable by the induction hypothesis.

  Suppose instead that $v_i \in S_i$.
  Then  there exists a vertex $v_j \in S_{i}$ that dominates $v_i$ in $G_i[S_i]$.
  To see why, let $Q$ be the neighbours of $v_i$ in $G_i[S_i]$.
  Because $|S_i|>1$ and $G_i[S_i]$ is connected, the set $Q \neq \emptyset$.
  Consider the following two cases:
  \begin{enumerate}
    \item If $Q$ is a clique, then any $v_j \in Q$ dominates $v_i$ in $G_i[S_i]$.
    \item If $Q$ is not a clique, then there exist two non-adjacent vertices $x,y\in Q$ and $(x,v_i,y)$ is a shortest path in $G_i$.
      Since $G_i$ is dismantlable, there exists some $v_j$ that dominates $v_i$ in $G_i$.
      Thus, $(x,v_j,y)$ is a shortest path from $x$ to $y$ in $G_i$.
      This implies that $v_j \in S_{i}$. Hence, $v_j$ also dominates $v_i$ in $G_i[S_i]$.
  \end{enumerate}
  We observe that $S_{i+1}$ is convex in $G_{i+1}$.
  Let $\rho$ be a shortest path between any two vertices in $S_{i+1}$.
  By \Cref{lemma:isometric}, the path $\rho$ is a shortest path in $G_i$. Since $S_i$ is convex, this path is contained in $S_i$.
  Thus, $\rho$ is also contained in $S_i \setminus \{v_i\} = S_{i+1}$.
  Since $S_{i+1}$ is convex, the induction hypothesis yields that  $G_{i+1}[S_{i+1}]$ is dismantlable. Since $v_i$ is dominated in $G_i[S_i]$ by $v_j \in S_{i+1}$, it follows that $G_i[S_i]$ is dismantlable.
\end{proof}

\end{document}